\documentclass[letterpaper, 11pt]{article}

\usepackage[margin=1in]{geometry}

\usepackage{amsthm}
\usepackage{amsmath}
\usepackage{algorithm}
\usepackage[noend]{algpseudocode}
\newtheorem{theorem}{Theorem}[section]
\newtheorem{corollary}[theorem]{Corollary}
\newtheorem{lemma}[theorem]{Lemma}

\usepackage{amsfonts}

\usepackage{comment}

\usepackage{graphicx}

\newcommand\epsi{\left({1}/{\epsilon}\right)^i}
\newcommand\eps[1]{\left(\frac{1}{\epsilon}\right)^{#1}}
\newcommand\nfrac{n^{1+\frac{1}{\kappa}}}
\newcommand\congest{{\sf CONGEST}}

\newcommand\congestmo{{\sf CONGEST} model}

\newcommand\centell{\left\lceil {\log \frac{\kappa+1}{2}}\right\rceil}
\newcommand{\degi}{n^\frac{2^i}{\kappa}}
\newcommand{\oeb}{(1+\epsilon,\beta)}
\newcommand{\betaemu}{O\left(\frac{{\log \kappa}}{\epsilon}\right)^{{\log \kappa}-1}}

\newcommand{\ize}{\lfloor{\log \kappa\rho}\rfloor}
\newcommand{\distell}{i_0+ \lceil \frac{\kappa+1}{\kappa\rho}\rceil-1}

\usepackage{hyperref}% http://ctan.org/pkg/hyperref
\usepackage{cleveref}% http://ctan.org/pkg/cleveref

\begin{document}
	
	\title{Ultra-Sparse Near-Additive Emulators\thanks{This research was supported by the ISF grant No. (2344/19).}}
	\renewcommand\footnotemark{}
	\author{Michael Elkin$^1$
		%	\footnote{This research was supported by the ISF grant No. (2344/19).}
		$\ $and Shaked Matar$^1$}

	\date{$^1$Department of Computer Science, Ben-Gurion University of the Negev, Beer-Sheva, Israel.\\
		Email: \texttt{elkinm@cs.bgu.ac.il, matars@post.bgu.ac.il}}

\maketitle

\begin{abstract}
	Near-additive (aka $(1+\epsilon,\beta)$-) emulators and spanners are a fundamental graph-algorithmic construct, with numerous applications for computing approximate shortest paths and related problems in distributed, streaming and dynamic settings.

	Known constructions of near-additive emulators enable one to trade between their sparsity (i.e., number of edges) and the additive stretch $\beta$. Specifically, for any pair of parameters $\epsilon >0$, $ \kappa=1,2,\dots$, one can have a $(1+\epsilon,\beta)$-emulator with $O(\nfrac)$ edges, with $\beta = \left(\frac{\log \kappa}{\epsilon}\right)^{\log \kappa}$. At their sparsest, these emulators employ $c\cdot n$ edges, for some constant $c\geq 2$. 
	We tighten this bound, and show that in fact precisely $\nfrac$ edges suffice.
	
	In particular, our emulators can be \emph{ultra-sparse}, i.e., we can have an emulator with $n+o(n)$ edges and $\beta = \left(\frac{\log {\log n}}{\epsilon }\right)^{{\log {\log n}}(1+o(1))}$. 
	
	We also devise a distributed deterministic algorithm in the \congestmo\ that builds these emulators in low polynomial time (i.e., in $O(n^\rho)$ time, for an arbitrarily small constant parameter $\rho >0$). 
	
	Finally, we also improve the state-of-the-art distributed deterministic \congest-model construction of 
	$(1+\epsilon,\beta)$-spanners devised in the PODC'19 paper 
	\cite{ElkinMatar}. Specifically, the spanners of \cite{ElkinMatar} have $O(\beta\cdot \nfrac)$ edges, i.e., at their sparsest they employ	
	 $ O\left(\frac{\log {\log n}}{\epsilon }\right)^{{\log {\log n}}}\cdot n$ edges. In this paper, we devise an efficient distributed deterministic \congest-model algorithm that builds such spanners with $O(\nfrac)$ edges for $\kappa = O\left(\frac{\log n}{\log ^{(3)}n}\right)$. At their sparsest, these spanners employ only $O(n\cdot {\log {\log n}})$ edges. 
\end{abstract}
	\maketitle

\newpage

\section{Introduction}

\subsection{Background and Our Results}

Given an unweighted undirected $n$-vertex graph $G=(V,E)$, and a pair of parameters $\alpha \geq 1$, $\beta \geq 0$, a
graph $G'=(V',E')$, with $V\subseteq V'$ is called an \emph{$(\alpha,\beta)$-emulator} for $G$, if for every pair of vertices $u,v\in V$ it holds that 
\begin{equation*}
d_{G}(u,v)\leq d_{G'}(u,v)\leq \alpha d_{G}(u,v)+\beta.
\end{equation*}

If $G'$ is a subgraph of $G$, it is called an \emph{$(\alpha,\beta)$-spanner}. 

In STOC'01, Elkin and Peleg \cite{ElkinP01} showed that for any $\epsilon >0$ and $\kappa =1,2,\dots$, there exists a $\beta = \beta(\epsilon,\kappa)$ such that for any $n$-vertex graph $G$ there exists a $(1+\epsilon,\beta)$-emulator of size $O({\log \kappa}\cdot \nfrac)$ and a $(1+\epsilon,\beta)$-spanner of size $O(\beta\cdot \nfrac)$. 
Emulators and spanners with these parameters are called \emph{near-additive}. The parameter $\beta$ is called the \emph{additive stretch} or \emph{error} of the respective emulator or spanner. 
In \cite{ElkinP01} the additive stretch is $\beta \approx \left(\frac{{\log \kappa}}{\epsilon}\right)^{{\log \kappa}-1}$, and this estimate stays the state-of-the-art. 
Based on \cite{AbboudB16}, Abboud et al. \cite{AbboudBP18} showed a lower bound of 
$\beta = \Omega \left(\frac{1}{\epsilon{\log \kappa}}\right)^{{\log \kappa}-1}$.

In SODA'06, Thorup and Zwick \cite{ThorupZ06} devised another scale-free construction of near-additive emulators. Their size and additive stretch are similar to those in \cite{ElkinP01}, but the same construction applies for all $\epsilon>0$.

Near-additive emulators and spanners were a subject of intensive research in the last two decades \cite{Elkin01,ElkinZ04,ThorupZ06,Pettie07,Pettie08,Pettie10,ElkinN16,ElkinN17,ElkinN20,ElkinP01,AbboudBP18,HuangP18,ElkinMatar}. They found numerous applications for computing almost shortest paths and distance oracles in various computational settings \cite{Elkin01,ElkinZ04,BernsteinR11,ElkinP15_distance_oracles,ElkinN17,AndoniSZ20}. 
Moreover, a strong connection between them and hopsets was discovered in \cite{ElkinN16,ElkinN17,HuangP17a}. Hopsets are also extremely useful for dynamic and distributed algorithms \cite{HenzingerKN15c,coh94,HenzingerKN16,LenzenP15,ElkinN16routing,ElkinN17routing,Censor-HillelDK19,DoryP20,LackiN20}. 
See also a recent survey \cite{ElkinN20} for an extensive discussion about the relationship between emulators, spanners and hopsets.

A significant research effort was put into decreasing the sparsity level of near-additive emulators and spanners \cite{Pettie07,Pettie08,Pettie10,ElkinN17,AbboudBP18}. Pettie \cite{Pettie08} showed that one can efficiently construct near-additive spanners of size $O(n\cdot ({\log{\log n} })^\phi)$, where $\phi \approx 1.44$ (with $\kappa = {\log n}$ and 
$\beta = \left(\frac{{\log {\log n}}}{\epsilon}\right) ^{\phi{\log {\log n}}}$). He then further improved the size bound to $O(n\cdot ({\log ^{(4)}n}))$, where ${\log ^{(4)}n}$ is the four-times iterated logarithm \cite{Pettie10}. 
The latter construction is however less efficient. (By efficient construction, we mean here centralized running time of $O(|E|\cdot n^\rho)$, for an arbitrarily small constant $\rho>0$, and a distributed \congest\ time of $O(n^\rho)$. We also call the latter \emph{low polynomial time}.)

Elkin and Neiman \cite{ElkinN17} devised an efficient construction of near-additive spanners with size $O(n{\log{\log n}})$ and of linear-size emulators. In the same paper, they also came up with an efficient distributed construction of \emph{ultra-sparse} (i.e., of size $n+o(n)$) multiplicative spanners\footnote{A subgraph $G'(V,H)$ is said to be a \emph{multiplicative $k$-spanner} of $G=(V,E)$ if for every pair of vertices $u,v\in V$, $d_{G'}(u,v)\leq k\cdot d_G(u,v)$. The ultra-sparse multiplicative spanners of \cite{ElkinN17} have stretch ${\log n}\cdot f(n)$ for an arbitrary slow-growing function $f(n) = \omega(1)$.}.

In the current paper we devise the first construction of \emph{ultra-sparse near-additive emulators}. Specifically, we show that for any $\epsilon >0$ and $\kappa=1,2,\dots$, there exists $\beta = \beta(\epsilon,\kappa)\approx \left(\frac{{\log \kappa}}{\epsilon}\right)^{\log \kappa-1}$, such that for any $n$-vertex graph $G=(V,E)$, there exists a $(1+\epsilon,\beta)$-emulator of size at most $\nfrac$. (Note that the leading constant in front of $\nfrac$ is $1$.)
By substituting here $\kappa = \omega({\log n})$, one obtains a near-additive emulator with $\beta = \left(\frac{\log{\log n}}{\epsilon}\right)^{{\log {\log n}}(1+o(1))}$ and size $n+o(n)$.

We also devise efficient (in the above sense) centralized and distributed deterministic algorithms that construct ultra-sparse emulators. 
Specifically, for any arbitrarily small constant $\rho >0$ (in addition to $\epsilon>0$ and $\kappa=1,2,\dots$), there exists $\beta= \beta(\epsilon,\kappa,\rho) $ such that our distributed deterministic \congest-model algorithm (see Section \ref{sec pre congest} for the definition of \congestmo) computes $\oeb$-emulators with at most $\nfrac$ edges, for 
$\beta = \left(\frac{{\log \kappa\rho}+1/\rho}{\epsilon\rho}\right)^{{\log \kappa\rho}+1/\rho}$ in time $O(\beta n^\rho)$. In particular, our algorithm can construct ultra-sparse emulators with $\beta = \left(\frac{\rho{\log{\log n}}+1/\rho}{\epsilon\rho} \right)^{{\log {\log n}}(1+o(1))}$ in deterministic distributed low polynomial time.

A variant of our algorithm also constructs sparse near-additive spanners. Specifically, the state-of-the-art distributed \congest-model deterministic algorithm for building near-additive spanners is due to \cite{ElkinMatar}. For any $\epsilon >0$, $\kappa =1,2,\dots$, $\rho>0$, there exists $\beta= \beta(\epsilon,\kappa,\rho) =
 \left(\frac{{\log \kappa\rho}+1/\rho}{\epsilon\rho}\right)^{{\log \kappa\rho}+1/\rho}$, such that there for any $n$-vertex graph $G= (V,E)$, the algorithm of \cite{ElkinMatar} constructs a $(1+\epsilon,\beta)$-spanner with $O(\beta\cdot\nfrac)$ edges in low polynomial time $O(\beta n^\rho)$. 
At their sparsest, the spanners of \cite{ElkinMatar} employ $n\cdot \left(\frac{{\log {\log n}}+1/\rho}{\epsilon\rho}\right)^{{\log {\log n}}+O(1)}$ edges. Improving upon this result, we devise a deterministic \congest-model algorithm with the same running time that constructs $\oeb$-spanners with $O({\log \kappa}\cdot \nfrac)$ edges. At their sparsest, these spanners employ just $O(n{\log {\log n}})$ edges. 

\subsection{Technical Overview}

All known  constructions of sparse near-additive emulators and spanners can be roughly divided into those that follow the \emph{superclustering-and-interconnection} (henceforth, SAI) approach of \cite{ElkinP01} and those that follow its scale-free version \cite{ThorupZ06}. (The constructions of \cite{Elkin01,ElkinZ04} follow a different approach, and result in emulators of size at least $\Omega(n{\log n})$.) 

In the SAI approach, one starts with a partition $P_0 = \{ \{v\}\ |\ v\in V \}$ of the vertex set into singleton clusters. Let $\ell \approx {\log \kappa}$. There are $\ell+1$ phases, numbered $0,1,\dots,\ell$, and in all phases except the last one there are two steps: the superclustering and the interconnection. In the last phase, the superclustering step is skipped. The input to each phase $i\in [0,\ell]$ is a partial partition\footnote{A family $A$ of pairwise disjoint subsets of a set $B$ 
	is called a \emph{partial partition} of $B$.}
 $P_i$ of $V$.
The phase also accepts as input two parameters, $\delta_i$ and $deg_i$, where the \emph{distance threshold} $\delta_i$ determines which clusters of $P_i$ are considered close or \emph{nearby} (those whose centers are at distance at most $\delta_i$ in $G$ from one another), and $deg_i$ determines how many nearby clusters a cluster needs to have to be considered \emph{popular} (at least $deg_i$). 

Intuitively, popular clusters $C$ then create \emph{superclusters} around them, which contain $C$ and all the nearby clusters. 
Unpopular clusters are \emph{interconnected} via emulator edges of weight equal to the distance between them. 
The set of superclusters is the partial partition $P_{i+1} $ for the $(i+1)$st phase.

In Thorup-Zwick's \cite{ThorupZ06} scale-free version of this construction, clusters of $P_i$ are sampled independently at random, with probability $\frac{1}{\deg_i}$ each, and each unsampled cluster joins the closest sampled cluster. In this way superclusters of $P_{i+1}$ are created. 
In addition, for every unsampled cluster $C$, it is connected via an emulator edge to every other unsampled cluster $C'$ which is closer to it than the closest sampled cluster. The weight of this edge is equal to the distance in $G$ between the respective cluster centers. This is an analogue of the interconnection step from \cite{ElkinP01}.

In both these approaches, ultimately the number of edges in the emulator is analyzed as the sum over all phases $i$ of the number of edges added to the emulator on phase $i$. 
One notes that each superclustering step forms a forest and thus contributes $O(n)$ edges. In addition, in \cite{ElkinP01}, the degree sequence $deg_i$ is designed in such a way that the interconnection step of each phase $i$ contributes at most $\nfrac$ edges. 
As a result, the overall size of the emulator is $O(({\log \kappa})(n+\nfrac)) = O({\log \kappa }\cdot \nfrac)$. For $\kappa = {\log n}$ this becomes $O(n\cdot {\log {\log n}})$. 
Subsequent improvements in the sparsity level of near-additive emulators and spanners \cite{Pettie08,Pettie10,ElkinN17,AbboudBP18,HuangP18,Pettie07} optimized the degree sequence $deg_0,deg_1,\dots,\deg_\ell$, so that the numbers of edges $m_0,m_1,\dots,m_\ell$ contributed on the interconnection steps of phases $0,1,\dots,\ell$, respectively, decrease geometrically and the total number of edges sums up to $O(\nfrac)$.

In this way one can guarantee that the overall contribution of interconnection steps is $O(\nfrac)$, while the additive stretch $\beta$ grows very little if at all. 
Elkin and Neiman \cite{ElkinN17} argued also that the overall contribution of superclustering steps is $O(n)$ (as opposed to the naive $O(n{\log \kappa})$), and as a result derived an emulator of linear size. 

Our main technical contribution is in a novel analysis. We adopt the original degree sequence of \cite{ElkinP01} as is, rather than using the optimized degree sequences from \cite{Pettie08,Pettie10,ElkinN17}. We then argue that the overall contribution of all the superclustering and interconnection steps \emph{together} is at most $\nfrac$. We achieve this by carefully charging edges inserted to the emulator during the entire algorithm to vertices, and arguing that no vertex is overloaded. By doing so we obtain a new structural understanding of this important construction, and derive the existence of ultra-sparse near-additive emulators.

We also show an efficient centralized implementation of this algorithm. Specifically, given parameters $\epsilon>0$, $\rho>0$ and $\kappa = 1,2,\dots,$ our algorithm constructs a $(1+\epsilon,\beta)$-emulator with at most $\nfrac$ edges and $\beta= \beta(\epsilon,\kappa,\rho) = \left(\frac{{\log \kappa\rho}+1/\rho}{\epsilon}\right)^{{\log \kappa\rho}+1/\rho}$ in $O(|E|\cdot n^\rho)$ deterministic time. This running time matches the state-of-the-art running time known for building denser near-additive emulators and spanners \cite{Pettie10,ElkinN17,ElkinMatar}.

\subsubsection{Distributed Implementation}

Our distributed algorithm is deterministic, and at the end of it, every vertex $v\in V$ knows about all emulator edges incident on it.

A large research effort was invested in implementing the SAI approach efficiently in the distributed \congestmo\ in the context of near-additive spanners and hopsets \cite{ElkinP01,ElkinN16,ElkinN17,HuangP17a,ElkinMatar}. Implementing this approach in the \congestmo\ in the context of near-additive \emph{emulators} presents a new challenge, since for every new emulator edge $e= (u,v)$, both its endpoints need to be aware of its existence and weight. Specifically, the center of every supercluster needs to learn of all clusters that have joined its supercluster. Since the number of clusters that join a single supercluster might be very large, this causes congestion. This issue does not arise in the construction of near-additive spanners, as spanner edges can be added locally, and the center of the supercluster does not need to learn of all the clusters that have joined its supercluster. In the context of hopsets, this challenge was addressed by broadcasting messages along a BFS tree that spanned the entire graph. This approach results in running time which is at least linear in the graph's diameter, which may be prohibitively large. 

In the current paper, we devise a superclustering scheme that ensures that the number of messages each vertex has to send in every step is relatively small. This is done by splitting very large superclusters into many superclusters. Intuitively, such a splitting may result in an increased number of levels $\ell$ of the construction, and therefore, in a higher additive term $\beta$ and increased running time. We show that this is not the case for our algorithm.

To our knowladge, there are no known distributed deterministic algorithms for building near-additive emulators of linear size. The only existing algorithm with these properties is the randomized algorithm of Elkin and Neiman \cite{ElkinN16}. However, the algorithm of \cite{ElkinN16} does not ensure that for every emulator edge $(u,v)$, both endpoints know of its existence. (Our algorithm provides, in fact, ultra-sparse emulators, while that of \cite{ElkinN16} guarantees just linear size.)

The only known distributed \congest\ \emph{deterministic} algorithm for building near-additive spanners \cite{ElkinMatar} constructs spanners of size $O(\beta \nfrac)$. The construction there can be adapted to build emulators of size $O({\log \kappa}\cdot \nfrac)$ (i.e., of size $\Omega(n\cdot {\log {\log n}})$), but, like the algorithm of \cite{ElkinN16}, it does not guarantee that for every emulator edge $e= (u,v)$ both $u$ and $v$ will be aware of it.

\subsection{Related Work}

The problem of efficiently constructing ultra-sparse multiplicative spanners and emulators was extensively studied in \cite{AlthoferDDJS93,HalperinZ96,RodittyZ04,DubhashiMPRS05,DerbelMZ06,Pettie07,Pettie10,ElkinN17}.

Spanners and emulators are known to be related to spectral sparsifiers \cite{KapralovP12,JambulapatiS20}. 
Ultra-sparse sparsifiers, or shortly ultra-sparsifiers, play a key role in a variety of efficient algorithms. Spielman and Teng \cite{SpielmanT04} used them for computing sequences of preconditioners. See also \cite{CohenKMPPRX14,KelnerLOS14,KollaMST10,Sherman13} for their applications to maximum flow and other fundamental problems.

We believe that the problem of devising ultra-sparse near-additive emulators is as fundamental as that of devising ultra-sparse multiplicative spanners and ultra-sparsifiers.

\subsection{Outline}
Section \ref{sec preliminaries} provides basic definitions for this paper.
In Section \ref{sec cent} we present a construction of ultra-sparse near-additive emulators in the centralized model. The properties of the construction are summarized in Corollaries \ref{coro emu} and \ref{coro emu us}.
In Section \ref{sec congest} we show a distributed \congest\ implementation of our algorithm. The properties of the distributed construction are summarized in Corollaries \ref{coro emu dist} and \ref{coro emu us dist}. In Section \ref{sec fast cent} we provide an efficient centralized construction of ultra-sparse near-additive emulators, which is based on our distributed construction. The properties of the construction are summarized in Theorems \ref{theo emu cent fast} and \ref{theo emu us cent fast}.
Section \ref{sec span} contains an efficient, deterministic \congest-model construction of sparse near-additive spanners. The properties of the construction are summarized in Corollary \ref{coro span}.

\subsection{Preliminaries }\label{sec preliminaries}

Throughout this paper, we denote by $r_C$ the center of the cluster $C$ and say that $C$ is \textit{centered around} $r_C$. The center $r_C$ is a designated vertex from $C$, i.e., $r_C\in C$. 
Throughout the paper, when the logarithm base is unspecified, it is equal to $2$. 
For a pair of integers $a,b$, where $ a\leq b$, the term $[a,b]$ stands for $\{a,a+1,\dots, b\}$.

\subsubsection{The Distributed \congest\ Model}	\label{sec pre congest}

In the distributed model \cite{Peleg00book} we have processors residing in vertices of the graph. The processors communicate with their graph neighbors in synchronous rounds. In the \congestmo, messages are limited to $O(1)$ words, i.e., $O(1)$ edge weights\footnote{In our case, the input graph is unweighted.} or ID numbers. %In the \localmo, the length of messages is unbounded. 
The running time of an algorithm in the distributed model is the worst case number of communication rounds that the algorithm requires.

\subsubsection{Ruling Sets}\label{sec def ruling setes}

Given a graph $G= (V,E)$, a set of vertices $W\subseteq V$ and parameters $\alpha,\beta \geq 0$, a set of vertices $A\subseteq W$ is said to be
an \textit{$(\alpha,\beta)$-ruling set} for $W$ if for every pair of vertices $u,v\in A$, the distance between them in $G$ is at least $\alpha$, and for every $u\in W$ there exists a representative $v\in A$ such that the distance between $u,v$ is at most $\beta$.

\section{Centralized Construction}\label{sec cent}

In this section we devise an algorithm that, given a graph $G=(V,E)$ on $n$ vertices and parameters $\epsilon<1$ and $\kappa \geq 2$ constructs a $(1+\epsilon,\beta)$-emulator for $G$ with at most $\nfrac$ edges in polynomial time in the centralized model, where $\beta = O\left(\frac{\log \kappa}{\epsilon}\right)^{{\log \kappa}-1}$.
In particular, by setting $\kappa = \omega({\log n})$, we construct an emulator with $n+o(n)$ edges with $\beta = \left(\frac{{\log{\log n}}}{\epsilon}\right)^{{\log {\log n}} + O(1)}$.

Section \ref{sec cent const} contains a general overview of the centralized construction. The properties of the resulting emulator and of the construction are analyzed in Section \ref{sec cent analysis of const}.

Our centralized construction is based on the centralized algorithm of Elkin and Peleg \cite{ElkinP01}. As was described in the introduction, both algorithms (of \cite{ElkinP01} and ours) follow the SAI approach to constructing near-additive emulators. 
There are, however, some important differences in both the algorithm and in its analysis. In the algorithm of \cite{ElkinP01} popular clusters $C$ (see Section \ref{sec preliminaries} for its definition) create superclusters $\widehat{C}$ around them, that contain only clusters that are close to the cluster $C$.
% and these superclusters are structured as \emph{stars} with $C$ becoming the star center and all nearby clusters $C'$ that are merged $\widehat{C}$ become the star leaves. 
All unpopular clusters $C''$ that are not merged into one of the superclusters are then interconnected with other nearby unpopular clusters, but  \emph{not} with nearby superclusters. To guarantee connectivity (and small stretch) between the superclusters and nearby unpopular clusters, the algorithm of \cite{ElkinP01} employs a separate ground partition. The spanning forest of this ground partition contributes at most $n-1$ additional edges to the emulator, which are swallowed by the overall size estimate of  $O(\nfrac)$. 

On the other hand, in our current algorithm we aim at a size bound of exactly $\nfrac$, and thus we cannot afford using a separate ground partition. Instead, once our algorithm creates a star-like supercluster $\widehat{C}$, it also inserts all  unclustered clusters $C''$  that are nearby the \emph{supercluster} $\widehat{C}$ into a set $N_i$ of \emph{buffer} clusters. These buffer clusters will not be allowed to create superclusters around them. They will be allowed to join other superclusters that will be constructed in future. However, if no future supercluster will incorporate them, the supercluster $\widehat{C}$ will do so. As a result, superclusters created in our algorithm may have larger radii in comparison to those constructed in \cite{ElkinP01}. 
This adaptation of the SAI approach of \cite{ElkinP01} takes care of the connectivity (and small stretch) between superclusters $\widehat{C}$ and their nearby unclustered clusters $C''$, \emph{without} paying an additional (additive) term of $n-1$ edges in the size of the emulator. 

In addition, the size analysis of \cite{ElkinP01} also analyzes separately contributions of different phases of the algorithm, and then sums them up. This is also the case in all subsequent works \cite{ThorupZ06,Pettie09,ElkinN17,ElkinMatar}.  As was discussed in the introduction, this naive summation (even with optimized degree sequences) is doomed to result in an emulator of size at least $\nfrac + n - O(1) \geq 2n-O(1)$. Our size analysis carefully combines contributions of all different phases altogether, and thus results in the bound of exactly $\nfrac$. 

Finally, the adaptation that we discussed above induces some modifications of stretch analysis as well. This is since, as discussed earlier, the radii of clusters constructed by our algorithm may be larger than the radii of clusters constructed in \cite{ElkinP01}. Specifically, both in our result and in that of \cite{ElkinP01}, $\beta = O\left( \frac{{\log \kappa}}{\epsilon}\right)^{{\log \kappa}-1}$, but the constant hidden by the $O$-notation in our bound is slightly larger than that in \cite{ElkinP01}. 

Yet another variant of the construction of \cite{ElkinP01} was given in \cite{ElkinN17}. In this variant of the construction, cluster centers are sampled, and clusters that are close to sampled clusters join them to create superclusters. As a result, the connectivity (and small stretch) between superclusters $\widehat{C}$ and nearby unclustered clusters $C''$ is ensured without employing a ground partition. On the other hand, this scheme requires randomization, while our approach is deterministic. Also, the size analysis of \cite{ElkinN17}, like that of \cite{ElkinP01}, analyzes each phase separately, and as a result, it cannot be used to provide ultra-sparse emulators.

\subsection{The Construction}\label{sec cent const}

Our algorithm initializes $H= \emptyset$ and proceeds in phases. The input to each phase $i\in [0,\ell]$ is a collection of clusters $P_i$, a degree parameter $deg_i$ and a distance threshold parameter $\delta_i$. 
The parameters $\ell,\{\deg_i,\delta_i\ | \ i\in[0,\ell]\}$ are specified in Section \ref{sec set param}. The set $P_0$ is initialized as the partition of $V$ into singleton clusters, i.e., clusters containing one single vertex each.

Consider an index $i\in [0,\ell]$, and let $C,C'$ be a pair of clusters in $P_i$, centered around vertices $r_C,r_{C'}$, respectively. We say that $r_C,r_{C'}$ are \emph{neighboring cluster centers} if $d_G(r_{C},r_{C'})\leq \delta_i$. If $r_C,r_{C'}$ are neighboring cluster centers, their respective clusters $C,C'$ are said to be \emph{neighboring clusters}.

Intuitively, in each phase $i$ the algorithm sequentially considers centers of clusters from $P_i$ and connects them with their neighboring cluster centers, i.e., it adds to the emulator $H$ an edge between them. The weight of each new edge is set to be the length of the shortest path in $G$ between its endpoints. Each added edge is charged to a center of a cluster in $P_i$. Centers of clusters that do not have many neighboring cluster centers are charged with all the edges that were added to the emulator when they were considered. However, cluster centers that have many neighboring cluster centers, i.e., \emph{popular cluster centers}, require a different approach. They are still connected with their neighboring cluster centers, but they are not charged with these edges. Instead, their neighbors are required to \textit{share the burden}.

Generally speaking, in each phase, we interconnect cluster centers that are not popular, and form superclusters around popular cluster centers. The set of superclusters formed in phase $i$ is the input $P_{i+1}$ for phase $i+1$. This allows us to defer work on these dense areas of the graph to later phases of the algorithm.

\subsubsection{Execution Details} 
We now describe the execution of a phase $i\in [0,\ell]$ of the algorithm.
At the beginning of phase $i$, define $S_i$ to be the set of centers of clusters $ C\in P_i $ and $U_i, N_i = \emptyset$. $U_i$ is the set of unclustered clusters during phase $i$. $N_i$ is an additional auxiliary set of cluster centers, which will be eventually superclustered. On the other hand, once a cluster $C$ joins $N_i$, it will not be allowed to create a supercluster around it.

The algorithm sequentially considers vertices from $S_i$.
While the set $S_i$ is not empty, the algorithm removes a single vertex $r_C$ from $S_i$. A Dijkstra exploration is executed from $r_C$ to depth $\delta_i$. 
Let $\Gamma(r_C)$ be the set of vertices $r_{C'}\in S_i\cup N_i$ that were discovered by the exploration (note that $r_C\notin S_i\cup N_i$, and so $r_C\notin \Gamma(r_C)$). For each vertex $r_{C'}\in \Gamma(r_C)$, the edge $(r_C,r_{C'})$ is added to the emulator $H$ with weight $d_G(r_C,r_{C'})$.

If $|\Gamma(r_C)|< deg_i$, then the center $r_C$ is charged with all edges added to $H$ as a result of an exploration originated from it. See Figure \ref{fig cent intercon} for an illustration. The cluster $C$ of the vertex $r_C$ is added to the set $U_i$ of unclustered clusters. 

%\begin{figure}
%	
%	\centering
%	\includegraphics[scale=0.07]{"cent intercon4".png}
%	\caption{\textbf{Interconnection edges.} The considered (dark gray) cluster center does not have many neighbors. The direction of the edges indicates that they are charged to the center of $C$.}
%	\label{fig cent intercon}
%\end{figure}

However, if $|\Gamma(r_C)|\geq deg_i$, then $r_C$ cannot be charged with these edges. A new supercluster $\widehat{C}$ is formed around $r_C$. The new supercluster contains the cluster $C\in P_i$ of $r_C$, and all clusters $C'$ such that their centers are in $\Gamma(r_C)$. The vertex $r_C$ becomes the center of the new supercluster $\widehat{C}$. The new supercluster $\widehat{C}$ joins the set $P_{i+1}$, which is the input collection for the next phase. See Figure \ref{fig cent supercluster} for an illustration. The cluster centers in $\Gamma(r_C)$ are removed from $S_i$ and from $N_i$. Thus, they will not be considered by the algorithm.

%\begin{figure}
%	\centering
%	\includegraphics[scale=0.08]{"cent superc3".png}
%	\caption{\textbf{Superclustering edges.} The considered cluster center (dark gray) has many neighbors. A supercluster $\widehat{C}$ is formed around $C$ and contains the neighbors $C'$ to which $C$ added an edge. The direction of the edges indicates that the white clusters are charged to the centers of the neighbors of $C$. }
%	\label{fig cent supercluster}
%\end{figure}
The algorithm described thus far is not sufficient. Consider a case where the algorithm has already formed some superclusters in phase $i$, and then it considers a cluster center $r_C$ that has at least $deg_i$ neighboring cluster centers, but many of them have been superclustered in this phase. See Figure \ref{fig clustering problem} for an illustration. The cluster center $r_C$ must be connected with its neighboring cluster centers. However, the center $r_C$ cannot be charged with these edges, nor can we form a supercluster around it containing all of its neighbors, as many of them already belong to other superclusters.

% \begin{figure}
% 	\centering
% 	\includegraphics[scale=0.13]{"clustering problem".png}
% 	\caption{ The cluster $C$ (depicted as white circle) can have many neighbors in $P_i$ that have joined superclusters in phase $i$ (the gray curved areas) before its center was considered. When we consider the center $r_C$ of $C$, we cannot connect it with its neighboring cluster centers, since we cannot charge $r_C$ for the added edges, nor can we require its neighbors to be charged for them. In the figure, the solid and dotted lines represent emulator and graph edges, respectively. }
% 	\label{fig clustering problem}
% \end{figure}

To avoid such occurrences altogether, when a supercluster $\widehat{C}$ is formed around a vertex $r_C$ in phase $i$, every cluster center $r_{C''}\in S_i$ with $\delta_i < d_G(r_C,r_{C''})\leq 2\delta_i$ is removed from $S_i$, and is added to a set $N_i$.
 
Consider a vertex $r_{C'}\in N_i$. If at the end of phase $i$ it has not been superclustered, it is added to the supercluster $\widehat{C}$ that was formed when $r_{C'}$ was added to $N_i$. Let $r_C$ be the center of $\widehat{C}$. The edge $(r_C,r_{C'})$ is added to $H$ with weight $d_G(r_C,r_{C'})$. This edge is charged to $r_{C'}$. 
See Figure \ref{fig cent supercluster2} for an illustration.

This completes the description of phase $i$.
Observe that the designation of a cluster center as popular or unpopular depends on the order in which the algorithm removes cluster centers from $S_i$. For example, consider the star graph $G=(V,E)$ with $V=\{u_0,u_1,\dots, u_n\}$ and $E=\{(u_0,u_i)| i\in [1,n] \}$. If in phase $0$ the algorithm begins by considering the cluster center $u_0$, then it is designated as popular. However, if the algorithm considers the cluster center $u_0$ last, then it will not be designated as a popular cluster, since it does not have any neighbors from $S_i\cup N_i$ (as at this point the sets $S_i,N_i$ are empty). Hence we cannot a-priori define a set of popular clusters.

\begin{figure}
	\begin{minipage}{.4\textwidth}
		\centering
		\includegraphics[scale=0.07]{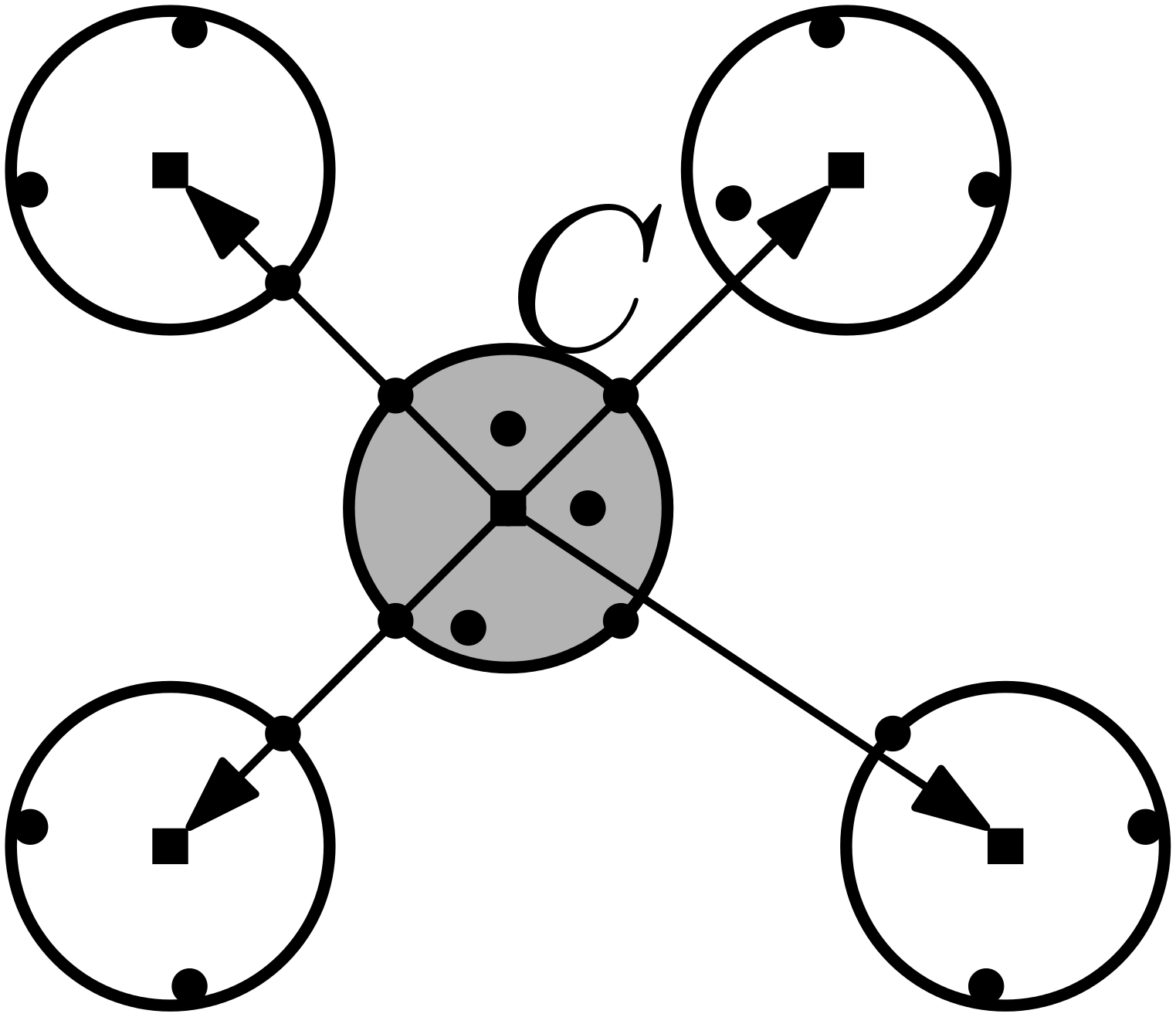}
		\caption{\textbf{Interconnection edges.} The considered (dark gray) cluster center does not have many neighbors. The direction of the edges indicates that they are charged to the center of $C$.}
		\label{fig cent intercon}
	\end{minipage}
	\begin{minipage}{.02\textwidth}
		\hspace{.05cm}
	\end{minipage}
	\begin{minipage}{.57\textwidth}
		\centering
		\includegraphics[scale=0.08]{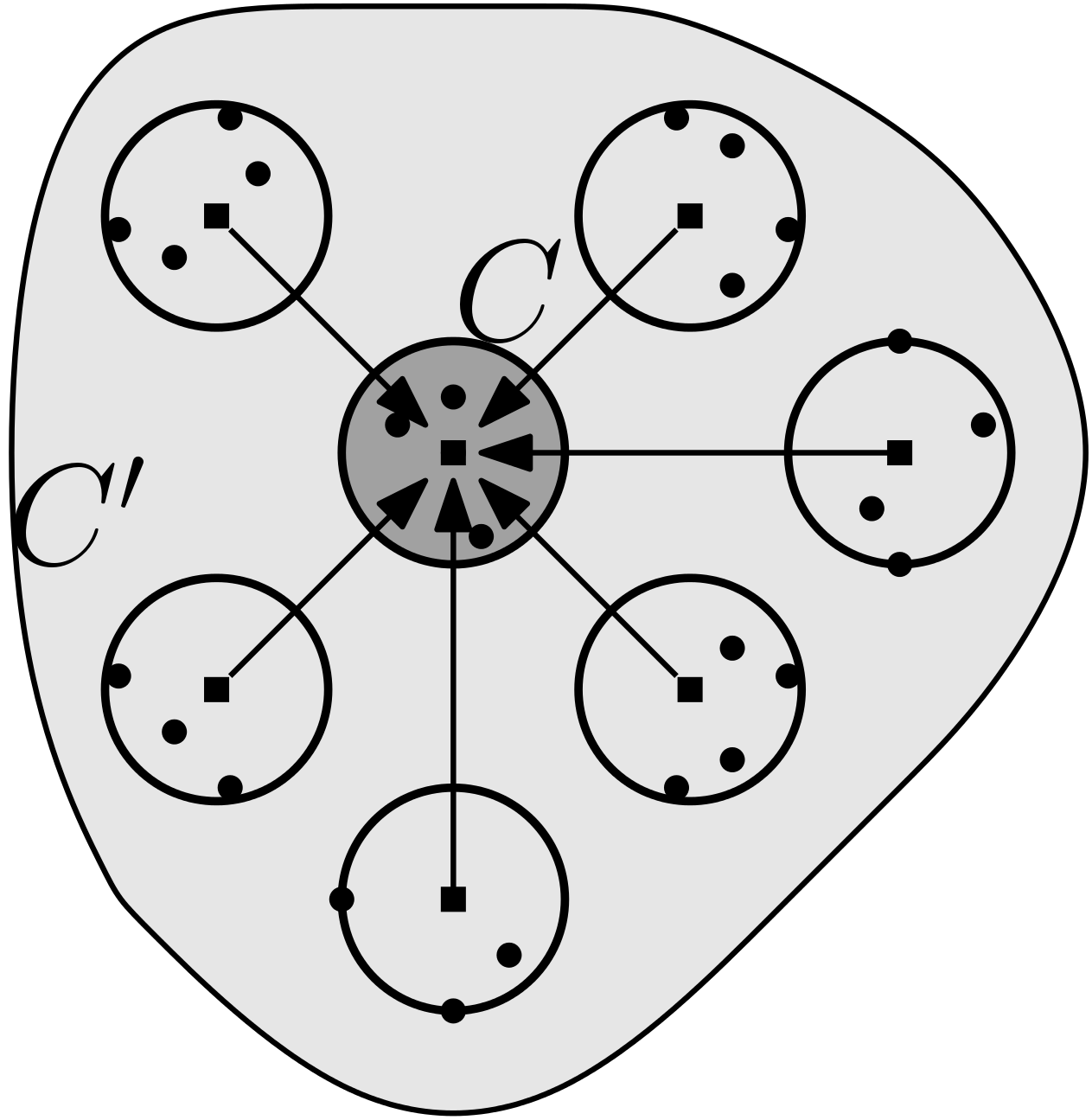}
		\caption{\textbf{Superclustering edges.} The considered cluster center (dark gray) has many neighbors. A supercluster $\widehat{C}$ is formed around $C$ and contains the neighbors $C'$ to which $C$ added an edge. The direction of the edges indicates that the white clusters are charged to the centers of the neighbors of $C$. }
		\label{fig cent supercluster}
	\end{minipage}%
\end{figure}
\begin{figure}
	\begin{minipage}{.48\textwidth}
		\centering
		\includegraphics[scale=0.13]{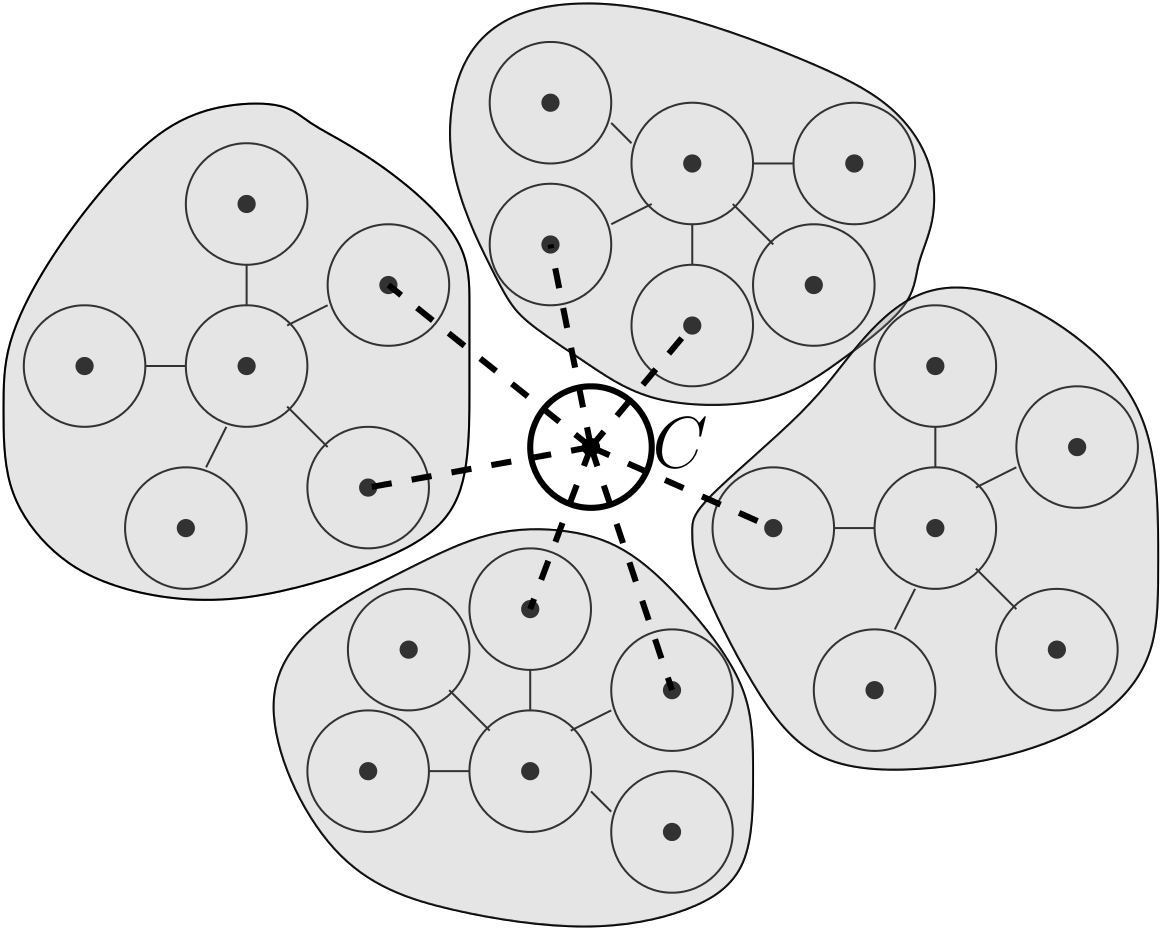}
		\caption{ The cluster $C$ (depicted as white circle) can have many neighbors in $P_i$ that have joined superclusters in phase $i$ (the gray curved areas) before its center was considered. When we consider the center $r_C$ of $C$, we cannot connect it with its neighboring cluster centers, since we cannot charge $r_C$ for the added edges, nor can we require its neighbors to be charged for them. In the figure, the solid and dotted lines represent emulator and graph edges, respectively. }
		\label{fig clustering problem}
	\end{minipage}
	\begin{minipage}{.02\textwidth}
		\hspace{.05cm}
	\end{minipage}
	\begin{minipage}{.48\textwidth}
		\centering
		\includegraphics[scale=0.07]{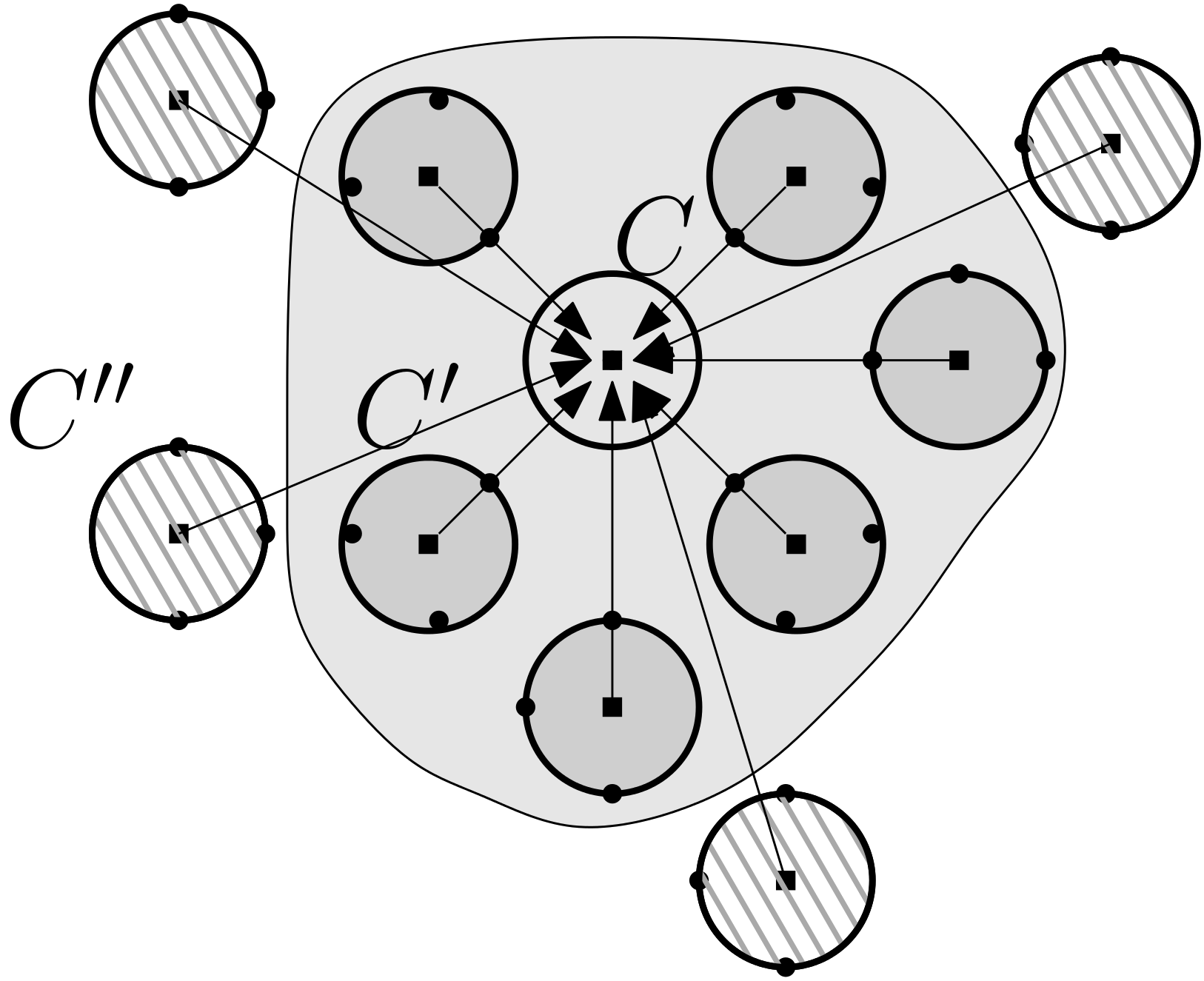}
		\caption{Clusters from $N_i$ join a neighboring supercluster. In the figure, the gray curved area represents a supercluster centered around $C$. The dark gray circles represent neighbors of $C$ that joined $\widehat{C}$ when it was formed. The stripped circles represent clusters that had centers in $S_i$ when $\widehat{C}$ was formed, and remained in $N_i$ until the end of phase $i$. At the end of phase $i$, they join the cluster $\widehat{C}$. Their cluster centers are charged with the edges that connect them with $\widehat{C}$. }
		\label{fig cent supercluster2}
	\end{minipage}%
\end{figure}
%This completes the description of the algorithm. 
%The pseudocode of the algorithm is given in Algorithm \ref{Alg process cluster cent}. 

%
%\begin{figure}
%		\centering
%		\includegraphics[scale=0.07]{"cent superc2".png}
%		\caption{Clusters from $N_i$ join a neighboring supercluster. In the figure, the gray curved area represents a supercluster centered around $C$. The dark gray circles represent neighbors of $C$ that joined $\widehat{C}$ when it was formed. The stripped circles represent clusters that had centers in $S_i$ when $\widehat{C}$ was formed, and remained in $N_i$ until the end of phase $i$. At the end of phase $i$, they join the cluster $\widehat{C}$. Their cluster centers are charged with the edges that connect them with $\widehat{C}$. }
%		\label{fig cent supercluster2}
%	\end{figure}
This completes the description of the algorithm. 
The pseudocode of the algorithm is given in Algorithm \ref{Alg process cluster cent}. 

\begin{algorithm}[H]
	\caption{Construction of a Near-Additive Emulator}
	\label{Alg process cluster cent}
	\begin{algorithmic}[1]
		\Statex \textbf{Input:} Graph $G=(V,E)$, a parameter $\ell\in \mathbb{N}$, and sequences $\langle deg_0,\dots,deg_\ell\rangle $, $\langle \delta_0,\dots,\delta_\ell\rangle $ 
		\State $P_{0} = \{ \{v\} \ | \ v\in V \}$
		\For {$i\in [0,\ell ]$}
		\State \label{line 3} $S_i \gets$ all centers of clusters from $P_i$
		\State \label{line 4} $U_i,N_i, P_{i+1} \gets \emptyset$
		\While {$S_i\neq \emptyset$}
		\State \label{line 6}remove a cluster center $r_C$ from $S_i$
		\For {all cluster centers $r_{C'}\in S_i\cup N_i$ s.t. $d_G(r_{C},r_{C'})\leq \delta_i$}\label{line 7}
		\State add to the emulator $H$ an edge $(r_{C},r_{C'}) $ with weight $d_G(r_{C},r_{C'})$\label{line 8}
		\EndFor
		\If {$r_C$ has less than $deg_i$ neighboring cluster centers in $ S_i\cup N_i$}\label{line 9}
		\State {add the cluster $C$ of $r_C$ to $U_i$}\label{line 10}
		\Else
		\State $\widehat{C}\gets C $ \label{line 12}
		\State $r_{\widehat{C}} \gets r_C$
		
		\For {all clusters $r_{C'}\in S_i\cup N_i$ such that $d_G(r_C,r_{C'})\leq \delta_i$}\label{line 14}
		\State remove $r_{C'}$ from $S_i\ or \ N_i$.\label{line 15}
		\State let $C'$ be the cluster centered at $r_{C'}$
		\State $\widehat{C}\gets \widehat{C} \cup \{C'\} $	
		\EndFor
		\For {all clusters $r_{C''}\in S_i$ such that $d_G(r_C,r_{C'})\leq 2\delta_i$}				
			\State $S_i = S_i \setminus \{r_{C''}\}$\label{line 17}
			\State $N_i = N_i \cup \{r_{C''}\}$\label{line 18}
		\EndFor
		\State $P_{i+1} \gets P_{i+1}\cup \{ \widehat{C}\}$\label{line 13}
		
		\EndIf
		\EndWhile
		\For{all cluster centers $r_{C''}\in N_i$} \label{line 19}
		
		\State let $\widehat{C}$ be the supercluster that was formed when $ r_{C''}$ joined $N_i$, and let $r_C$ be the center of $\widehat{C}$
		\State let $C''$ be the cluster centered at $r_{C''}$
		\State add to the emulator $H$ an edge $(r_{C},r_{C''}) $ with weight $d_G(r_{C},r_{C''})$\
		\State $\widehat{C}\gets \widehat{C}\cup \{C''\}$ 
		\EndFor
		\EndFor
	\end{algorithmic}
\end{algorithm}

\subsubsection{Setting Parameters}\label{sec set param}
In this section we specify the selection of the parameters $deg_i,\ell$ and $\delta_i$. 

The degree parameter $deg_i$ controls the number of edges added to the emulator, and also the number of phases required until we are left with a small number of clusters. 
For every $i\in [0,\ell]$, we set $deg_i=n^\frac{2^i}{\kappa}$.

Set $\ell = \centell$. In Section \ref{sec cent analysis of size} we show that $|P_{\ell}|\leq deg_\ell$. Therefore, there are no popular cluster centers in phase $\ell$, and superclusters are not formed during this phase. It follows that $P_{\ell+1}= \emptyset$ and $U_\ell = P_{\ell}$.

Define recursively $R_0= 0$, and for every $i\in [1,\ell]$ define $R_{i+1} = 2\delta_i+R_i$. The distance threshold parameter is defined by $\delta_i = \epsi+2R_i$, for every $i\in [0,\ell]$. Intuitively, $R_i$ is an upper bound on the radii of clusters in $P_i$, i.e., the maximal distance in the emulator $H$ between a center $r_C$ of a cluster $C\in P_i$ and a vertex $u\in C$. In Lemma \ref{lemma radpi leq ri} we prove that this inequality indeed holds. 
\subsection{Analysis of the Construction}
\label{sec cent analysis of const}
In Section \ref{sec cent analysis of size} we analyze the size of the emulator. 
In Section \ref{sec cent rt} we show that the algorithm can by executed in polynomial time. 
Finally, in Section \ref{sec cent stretch} we analyze the stretch of the emulator.

\subsubsection{Analysis of the Number of Edges}\label{sec cent analysis of size}

In this section, we analyze the size of the emulator $H$. We will charge each edge in the emulator $H$ to a single vertex.
We begin by proving that in the concluding phase $\ell$ there are no popular clusters. To do so we show that the size of $P_{\ell}$ is at most $deg_\ell$.

\begin{lemma}\label{lemma widehat} 
	For every index $i\in [0,\ell]$, each supercluster $\widehat{C}$ constructed in phase $i$ consists of at least $deg_i+1$ clusters from $P_i$. 
\end{lemma}

\begin{proof}
	Let $i\in[0,\ell]$, and let $\widehat{C}$ be a supercluster that was created around a cluster $C$ in phase $i$. The algorithm added edges from $C$ to at least $deg_i$ clusters from $S_i\cup N_i$. These clusters, and $C$ itself, all became superclustered into $\widehat{C}$. Thus, $\widehat{C}$ contains at least $deg_i+1$ clusters from $P_i$. 
\end{proof}

In the next lemma, we argue that superclusters are disjoint, and thus Lemma \ref{lemma widehat} can be used to bound the number of superclusters formed during each phase.

\begin{lemma}\label{lemma disjoint}
	For $i\in [0,\ell-1]$, all superclusters formed during phase $i$ are pairwise disjoint.
\end{lemma}

\begin{proof}	
	Let $\widehat{C}$ be a supercluster formed during phase $i$. 
	Recall that all centers of clusters that have joined $\widehat{C}$ were in $S_i\cup N_i$ until they joined $\widehat{C}$. Also recall that once a cluster joins a supercluster, its center is removed from $S_i$ and from $N_i$, and it is not added to $S_i$ or $N_i$ in the future.
	 
	On the one hand, this implies that all clusters that have joined $\widehat{C}$ did not join any other supercluster before they joined $\widehat{C}$. On the other hand, once a cluster joined $\widehat{C}$, its center is removed from $S_i$ and $N_i$, and therefore it will not join another supercluster in future. Hence, we conclude that all superclusters formed during phase $i$ are pairwise disjoint.
\end{proof}

 In the next lemma we provide an upper bound in the size of $P_i$, for every index $i\in [0,\ell]$.

\begin{lemma}\label{lemma bound pi}
	For $i\in [0,\ell]$, we have 
	$\quad|P_i| \leq n^{1-\frac{2^i-1}{\kappa}}.$
\end{lemma}

\begin{proof}
	The proof is by induction on the index $i$. 
	For $i=0$, the right-hand side of the equation is equal to $n$. Thus the claim is trivial. 
	
	For the induction step, assume that the claim holds for some $i\in [0,\ell-1]$. By Lemmas \ref{lemma widehat} and \ref{lemma disjoint} and the induction hypothesis, we obtain:
	\begin{equation*}
	\begin{array}{lclclclcl}
	|P_{i+1}|\leq n^{1-\frac{2^i-1}{\kappa}}\cdot (deg_i+1)^{-1}\leq n^{1-\frac{2^i-1}{\kappa}}\cdot n^{-\frac{2^i}{\kappa}}
	\leq n^{1-\frac{2^{i+1}-1}{\kappa}}.
	\end{array}
	\end{equation*}
	Hence the claim holds also for $i+1$. 
\end{proof}

Recall that $\ell = \centell$. Observe that Lemma \ref{lemma bound pi} implies that
\begin{equation}\label{eq degl pl cent}
\begin{array}{lclclclclclclclc}
|P_{\ell}|\leq n^{1-\frac{2^\ell-1}{\kappa}} \leq n^\frac{2^\ell}{\kappa} = deg_{\ell} .
\end{array}
\end{equation}

Therefore, in phase $\ell$ there are no popular clusters. It follows that $P_{\ell+1}$ is an empty set, and that $U_\ell = P_\ell$.

Next, we examine the edges added by each phase $i$ of the algorithm, and charge each edge to 
a center of a cluster $C\in P_i$. 
Recall that there are two types of edges in the emulator:
\begin{enumerate}
	\item \textit{Interconnection edges}, added when the algorithm %Algorithm \ref{Alg process cluster cent} 
	considered an unpopular cluster center $r_C$. These edges are charged to $r_C$. See Figure \ref{fig cent intercon}.
	
	\item \textit{Superclustering edges}, added when a cluster ${C'}$ joined a supercluster $\widehat{C}$ that was formed around a cluster $C$, where $C\neq C'$. See Figures \ref{fig cent supercluster}, \ref{fig cent supercluster2} for an illustration.
	These edges are charged to the centers of clusters $C'$ that were superclustered into the new supercluster formed around $C$. For example, if for some $h\geq 1$, clusters $C_1,C_2,\dots,C_h$, centered at vertices $v_1,v_2,\dots,v_h$, respectively, are clustered into a supercluster rooted at a cluster $C$, then each of these centers $v_1,v_2,\dots,v_h$ is charged with a single edge. Note that the center of the cluster $C$ is not charged with any edges in phase $i$. 
\end{enumerate}

Interconnection edges that were added in phase $i$ are charged to centers of clusters $C\in U_i$. Observe that a cluster $C$ has joined $U_i$ only if its center has added less than $deg_i$ edges to the emulator $H$. Therefore, phase $i$ adds at most $|U_i|\cdot deg_i$ interconnection edges to the emulator $H$. 
(Note that $U_i$ might be empty.)

Superclustering edges that were added in phase $i$ are charged to centers of clusters that did not join $U_i$, and also that no supercluster was formed around them in phase $i$. Thus, phase $i$ adds exactly $|P_i|-|U_i|-|P_{i+1}|$ superclustering edges. 

Hence, in phase $i\in [0,\ell]$, the number of edges added to the emulator $H$ is at most:
\begin{equation}\label{eq bound number i}
|U_i|\cdot deg_i + |P_i|-|U_i|-|P_{i+1}| = |P_i|+|U_i|\cdot (deg_i-1)-|P_{i+1}|. 
\end{equation}
In particular, this bound applies to the last phase $i=\ell$. Recall that by \cref{eq degl pl cent} we have $|P_{\ell}|= |U_\ell|\leq deg_\ell$, and also $|P_{\ell+1}| = \emptyset$. Therefore the bound becomes just $|P_{\ell}|\cdot deg_\ell$. 

We will now use the size of $P_{i+1}$ to bound the size of $U_i$. 
Observe that by Lemma \ref{lemma widehat} and because superclusters of $P_{i+1}$ are disjoint, we have that for all $i\in [0,\ell]$,
\begin{equation}
\label{eq bound ui cent}
|U_i|\leq |P_i|-|P_{i+1}|\cdot (deg_i +1).
\end{equation}

By \cref{eq bound number i,eq bound ui cent} we have that in phase $i$, the number of edges added to the emulator $H$ is at most:
\begin{equation}\label{eq bound number i2}
\begin{array}{rlclclclc}

	 & |P_i|+|U_i|\cdot (deg_i-1)-|P_{i+1}| \\
\leq & |P_i|+(|P_i|-|P_{i+1}|\cdot (deg_i +1))\cdot (deg_i-1)-|P_{i+1}|\\
=& |P_i|\cdot deg_i-|P_{i+1}|\cdot (deg_i^2 - 1)-|P_{i+1}|\\
=& |P_i|\cdot deg_i-|P_{i+1}|\cdot deg_i^2 .
\end{array}
\end{equation}

We are now ready to bound the size of the emulator $H$.
\begin{lemma}
	\label{lemma bound size cent all}
	The number of edges in the emulator $H$ satisfies 
	$|H| \leq
	\nfrac .$
\end{lemma}
\begin{proof}
	By \cref{eq bound number i2}, and since $P_{\ell+1}$ is an empty set, we obtain that the number of edges added by all phases of the algorithm is at most:
	\begin{equation*}
	\begin{array}{rllclclclclclc}
	%	\sum\limits_{i=0}^{\ell} (|P_i|+|U_i|\cdot (deg_i-1)-|P_{i+1}| )
	%	
	%	\leq& 
	\sum\limits_{i=0}^{\ell}(|P_i|\cdot deg_i-|P_{i+1}|\cdot deg_i^2)
	
	=& 
	
	|P_{0}|\cdot deg_0 + \sum\limits_{i=1}^{\ell} |P_i|\cdot (deg_i - deg^2_{i-1}).
	\end{array}
	\end{equation*}
	
	Recall that for all $i\in [0,\ell]$ we have $deg_i = n^\frac{2^i}{\kappa}$ thus $deg_i-deg_{i-1}^2 = 0$. Also, recall that $|P_{0}| = n$. Thus, the number of edges added to the emulator by all phases $i\in [0,\ell]$ is at most $\nfrac.$
\end{proof}

\subsubsection{Analysis of the Stretch}\label{sec cent stretch}

In this section we analyze the stretch of the emulator $H$. 
We begin by providing an upper bound on the radii of clusters in $P_i$.

For an index $i\in [0,\ell]$ and a cluster $C\in P_i$ centered around a vertex $r_C$, the radius of $C$ is defined to be $Rad(C) = { \max\{ d_H(r_C,v) \ | \ v\in C \} }$. The radius of the collection of clusters $P_i$ is defined to be $Rad(P_i) = { \max\{ Rad(C) \ | \ C\in P_i \} }$. We begin by proving that $R_i$ is an upper bound on the radii of clusters in $P_i$, for all $i\in [0,\ell]$. Recall that $\delta_i = \epsi+2R_i$, for every $i\in [0,\ell]$. Also, recall that 
$R_0= 0$, and for every $i\in [1,\ell]$, we have $R_{i+1} = 2\delta_i+R_i $.

\begin{lemma}\label{lemma radpi leq ri}
	For every index $i\in [0,\ell]$, we have $Rad(P_i)\leq R_i$.
\end{lemma}
\begin{proof}
	The proof is by induction on the index of the phase $i$. For $i=0$, all clusters in $P_i$ are singletons, and also $R_0=0$, and so the claim holds. 
	
	Assume the claim holds for some index $i\in [0,\ell-1]$ and prove that it holds for $i+1$. 
	Consider a cluster $\widehat{C}\in P_{i+1}$. This cluster was formed around a vertex $r_C$ during phase $i$. Consider a vertex $u\in \widehat{C}$. 
	
	\textbf{Case 1:} The vertex $u$ belonged to the cluster of $r_C$ in $P_i$. In this case, by the induction hypothesis we have $d_H(r_C,u)\leq R_i\leq R_{i+1}$.

	\textbf{Case 2:} The vertex $u$ belonged to a cluster $C'\in P_i$, where $r_C\notin C'$. Denote by $r_{C'}$ the center of the cluster $C'$.
	Since the center $r_{C'}$ joined the supercluster of $r_C$, we conclude that $d_G(r_C,r_{C'})\leq 2\delta_i$. When $r_{C'}$ joined the supercluster $\widehat{C}$, the edge $(r_C,r_{C'})$ was added to the emulator $H$, with weight $d_G(r_C,r_{C'})$. Thus, $d_H(r_C,r_{C'}) \leq 2\delta_i$. By the induction hypothesis, we also have $d_H(r_{C'},u)\leq R_i$. Hence, 
	$$d_H(r_C,u)\leq d_H(r_C,r_{C'})+d_H(r_{C'},u)\leq 2\delta_i+R_i = R_{i+1}.$$
\end{proof}

We now provide an upper bound on $R_i$. 
Observe that for every $i\in [1,\ell]$, we have $R_{i+1} = 2\delta_i+R_i = 2\epsi+5R_i$. 

\begin{lemma}
	\label{lemma bound ri}
	For every index $i\in [0,\ell]$, we have 
	$$R_i = {2}\cdot\sum_{j=0}^{i-1} \epsilon^{-j}\cdot 5^{i-1-j}.$$
\end{lemma}
\begin{proof}
The proof is by induction on the index $i$. For $i=0$, both sides of the equation are equal to $0$. So the base case holds. 

Assume that the claim holds for some index $i\in [0,\ell-1]$, and prove that it holds for $i+1$. 
By definition and the induction hypothesis we have: 
\begin{equation*}
\begin{array}{clllll}
R_{i+1} =& 2 \epsi+5\cdot
{2}\cdot\sum_{j=0}^{i-1}
\epsilon^{-j}\cdot 5^{i-1-j}
%\\
%&=& 2 \cdot \epsi+\left(
% 	{2}\cdot\sum_{j=0}^{i-1}
% 	\epsilon^{-j}\cdot 5^{i-j}\right)\\
=& 
{2}\cdot\sum_{j=0}^{i}
\epsilon^{-j}\cdot 5^{i-j}
\end{array}
\end{equation*}
\end{proof} 

By Lemma \ref{lemma bound ri}, we derive the following explicit bound on $R_i$, for all $i\in [0,\ell]$. 
\begin{equation*}
\begin{array}{clclclclclclc}
R_i =
2\cdot 5^{i-1}\sum_{j=0}^{i-1} (5\epsilon)^{-j}\leq
2\cdot 5^{i-1}\cdot \left(\frac{1}{5\epsilon}\right)^{i-1}\cdot \left(\frac{1}{1-5\epsilon}\right) = \frac{2}{1-5\epsilon}\cdot \eps{i-1}.
\end{array}
\end{equation*}

Assume that $\epsilon \leq 1/10$. It follows that 
\begin{equation}
\label{eq exp ri}
R_i \leq 4 \eps{i-1}.
\end{equation}

Next, we show that the emulator $H$ contains edges that connect every center of a cluster in $U_i$ with all its neighboring cluster centers. 
\begin{lemma}\label{lemma neighboring clusters cent}
	Let $i\in [0,\ell]$ and let $r_C$ be a center of a cluster $C\in U_i$. Then, for every neighboring cluster center $r_{C'}$ of $r_C$, we have $$d_H(r_C,r_{C'})=d_G(r_C,r_{C'}).$$
\end{lemma}

\begin{proof}
	Let $r_C$ be a center of a cluster $C\in U_i$ and let $r_{C'}$ be a neighboring cluster center of $r_C$, such that $r_{C'}\in C'$ and $C'\in P_i$. By definition, $d_G(r_C,r_{C'})\leq \delta_i$.
	Since $C\in U_i$, the algorithm has considered $r_C$, and executed a Dijkstra exploration from it during phase $i$.

	\textbf{Case 1:} The cluster center $r_{C'}$ was in $S_i\cup N_i$ when the cluster center $r_C$ was considered by the algorithm. 
	In this case, since $d_G(r_C,r_{C'})\leq \delta_i$,
	the edge $(r_C,r_{C'})$ was added to $H$ with weight $d_G(r_C,r_{C'})$.

	\textbf{Case 2:} The cluster center $r_{C'}$ was not in $S_i\cup N_i$ when the cluster center $r_C$ was considered by the algorithm. 
	Then, the cluster $C'$ of $r_{C'}$ has either joined $U_i$ or became superclustered before the cluster center $r_C$ was considered by the algorithm. Assume towards contradiction that ${C'}$ became superclustered before $r_C$ was considered. Therefore, a supercluster was grown around a cluster center $r_{C''}$ with $d_G(r_{C''},r_C)\leq 2\delta_i$, and so $r_C$ was removed from $S_i$, contradiction (see \cref{line 17} of Algorithm \ref{Alg process cluster cent}). Therefore, we conclude that $C'$ has joined $U_i$ before $r_C$ was considered by the algorithm. The algorithm has executed a Dijkstra exploration from $r_{C'}$, and since $d_G(r_C,r_{C'})\leq \delta_i$, the emulator $H$ contains the edge $(r_C,r_{C'})$ with weight $d_G(r_C,r_{C'})$. 
\end{proof}

We now show that for every vertex $v\in V$ there exists an index $i\in [0,\ell]$ such that $v$ belongs to a cluster that joins the set $U_i$. For notational purposes, define $U_{-1} = \emptyset$, and $U^{(i)} = \bigcup_{j=-1}^i U_i$ for all $i\in [-1,\ell]$. We say that a vertex $v$ is $U^{(i)}$-\emph{clustered} for some $i\in [0, \ell]$ if there exists a cluster $C \in U^{(i)}$ such that $v \in C$.  

\begin{lemma}
	\label{lemma ui partition}
	For every index $i\in [0,\ell]$, the set $P_{i} \cup U^{(i-1)}$ is a partition of $\ V$. 	
\end{lemma}

\begin{proof}
	The proof is by induction on the index of the phase $i$. For $i=0$, the claim is trivial since $P_{0}$ is a partition of $V$ into singleton clusters.
	
	Assume the claim holds for some index $i\in [0,\ell-1]$. Let $v\in V$. By the induction hypothesis, $v$ belongs to a cluster $C\in P_{i} \cup U^{(i-1)}$. If $C\in U^{(i-1)}$, then, by definition, $C\in U^{(i)} $. If $C\in P_i$, then in phase $i$, the cluster $C$ has either been superclustered into a supercluster of $P_{i+1}$, or it has joined $U_i$. In any case, $C\in P_{i+1} \cup U^{(i)}$, and so $P_{i+1} \cup U^{(i)}$ is a partition of $V$. Thus the claim holds for $i+1$. 
\end{proof}

Recall that by \cref{eq degl pl cent}, we have that $|P_{\ell}|\leq deg_\ell$ and thus, $P_{\ell+1} = \emptyset$. Therefore, Lemma \ref{lemma ui partition} implies that $U^{(\ell)}$ is a partition of $V$. 

In the following lemma we argue that superclusters form a laminar family. 

\begin{lemma} 
	\label{lemma laminar}
	Let $0 \leq j \leq i \leq \ell$ be a pair of indices. Let $C\in U_i$ be a cluster, and $v\in C$ be a vertex. Then, there exists a cluster $C'\in P_j$ such that $v\in C'$. 
\end{lemma}

\begin{proof}
	The proof is by induction on $i-j$. The induction base is $i-j=0$. Then, $C'=C$, and we are done.
	
	For the induction step, suppose that the assertion holds for some non-negative integer $h$. We prove it for $h+1$. 
	
	By the induction hypothesis, there exists a cluster $\tilde{C} \in P_{i-h}$ such that $v\in \tilde{C}$. Also, the cluster $\tilde{C}$ is a disjoint union of clusters from $P_{(i-h)-1}$. Hence, in particular, there exists a cluster $C'\in P_{i-h-1} = P_{i-(h+1)}$ such that $v\in C'$.
\end{proof}

%==============

We are now ready to bound the stretch of the emulator $H$. 
The outline of the proof is as follows. Consider a pair of vertices $u,v\in V$ and let $\pi(u,v)$ be the shortest path between them in $G$. We will show that $\pi(u,v)$ can be divided into smaller segments, and that for each such segment there is a path in $H$ between its endpoints $u',v'$ that is not significantly longer than the distance between $u',v'$ in the original graph $G$. 

Define recursively $\beta_0 = 0,\alpha_0=1$, and for $i>1$ define $\beta_i = 2\beta_{i-1}+6R_i$ and $\alpha_i = \alpha_{i-1} +\frac{\epsilon^i}{1-\epsilon^i}\cdot \beta_i$. 

\begin{lemma}\label{lemma emu stretch}
	Let $u,v\in V$ be a pair of vertices and let $\pi(u,v)$ be a shortest path between them. Let $i$ be the minimal index such that all vertices on the path $\pi(u,v)$ are $U^{(i)}$-clustered. Then, 
	$$d_H(u,v)\leq \alpha_i\cdot d_G(u,v)+ \beta_i.$$ 
\end{lemma}

\begin{proof}
	The proof is by induction on the index $i$. For $i=0$, all vertices on the path $\pi(u,v)$ are $U_0$-clustered, thus they all added to the emulator $H$ all edges that are incident to them.(Recall that $\delta_0 = 1/\epsilon^0 +2R_0 = 1$. See also \cref{line 7,line 8,line 9,line 10} of Algorithm \ref{Alg process cluster cent}.) Therefore, the path $\pi(u,v)$ itself is contained in the emulator $H$. 
	
	Let $i\in [1,\ell]$. Assume that the claim holds for $i-1$, and prove that it holds for $i$. 
	Let $(u,v)$ be a pair of vertices such that all vertices on a shortest path $\pi(u,v)$ are $U^{(i)}$-clustered. Denote $d= |\pi(u,v)|$. For convenience, we imagine that the vertices of $\pi(u,v)$ appear from left to right, where $u$ is the leftmost vertex and $v$ is the rightmost vertex.
	
	We divide the path 	$\pi(u,v)$ into segments $S_1,S_2,\dots S_q$, each of length exactly $\lfloor\epsi\rfloor$, except for the last segment that can be shorter than $\lfloor\epsi\rfloor$. Hence, $q\leq \lceil \frac{d}{\lfloor\epsi\rfloor} \rceil \leq \lceil \frac{d}{\epsi-1} \rceil = \lceil \frac{d\epsilon^i}{1-\epsilon^i} \rceil$.
	Consider a single segment $S$. Denote by $x,y$ the left and the right endpoints of $S$, and denote by $S=\pi(x,y)$ the subpath of $\pi(u,v)$ between them. (It is convenient to visualize the path $\pi(u,v)$ as going from the leftmost vertex $u$ to the rightmost vertex $v$.)
	
	\textbf{Case 1:} The segment $S$ does not contain a $U_i$-clustered vertex. Then, all the vertices of the segment are $U^{(i-1)}$ clustered. Hence, by the induction hypothesis 
	\begin{equation*}
		d_H(x,y)\leq
		 \alpha_{i-1}\cdot d_G(x,y)+ \beta_{i-1}.
	\end{equation*}

	\textbf{Case 2:} Let $z_1 $, $z_2 $ be the first and the last $U_i$-clustered vertices on the path $\pi(x,y)$, respectively. 
	Let $C_1,C_2\in U_i$ be the clusters such that $z_1\in C_1$ and $z_2\in C_2$. (Note that it is possible that $C_1=C_2$.)
	Both clusters intersect $S= \pi(x,y)$, and the length of $S$ is at most $\epsi$. In addition, by Lemma \ref{lemma radpi leq ri} we have that the radii of the clusters $C_1,C_2$ is at most $R_i$. Let $r_1,r_2$ denote the centers of clusters $C_1,C_2$, respectively. It follows that $d_G(r_1,r_2)\leq \epsi+2R_i = \delta_i$. Hence clusters $C_1,C_2$ are neighboring. 
	By Lemma \ref{lemma neighboring clusters cent} we have that $d_H(r_1,r_2)=d_G(r_1,r_2)$. See Figure \ref{fig emu path} for an illustration. As a result, by triangle inequality, 	
	\begin{equation*}
	d_H(r_1,r_2) = d_G(r_1,r_2) \leq R_i+ d_G(z_1,z_2)+R_i.
	\label{eq centers dist}
	\end{equation*}

	\begin{figure}
		\begin{minipage}{.48\textwidth}
			\centering
				\includegraphics[scale=0.15]{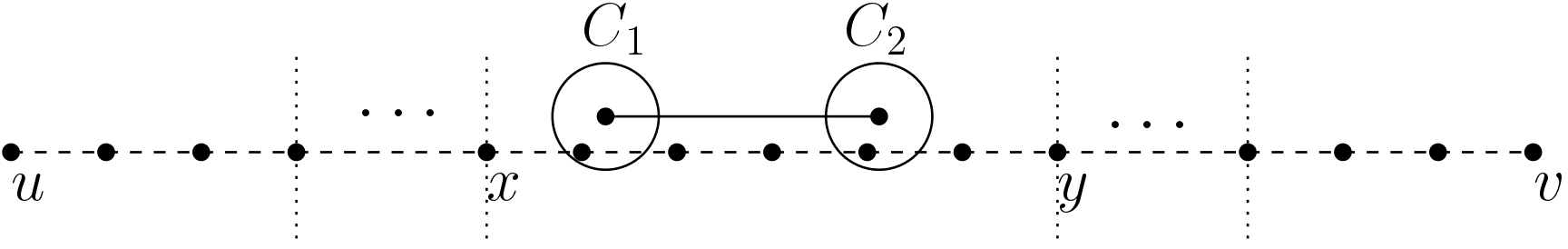}
			\caption{The analysis of the stretch of the emulator. For a pair $u,v\in V$ such that a shortest path $\pi(u,v)$ between them is $U^{(i)}$-clustered, we divide the path $\pi(u,v)$ into segments of length at most $\epsi$. For a single segment, denote by $C_1,C_2$ the first and the last $U_i$-clustered clusters on the path $\pi(u,v)$. The emulator $H$ contains an edge between the centers $r_1,r_2$ of the clusters $C_1,C_2$ with weight $d_G(r_1,r_2)$.} 		
			\label{fig emu path}
		\end{minipage}
		\begin{minipage}{.02\textwidth}
			\hspace{.05cm}
		\end{minipage}
		\begin{minipage}{.48\textwidth}
			\centering
			\includegraphics[scale=0.17]{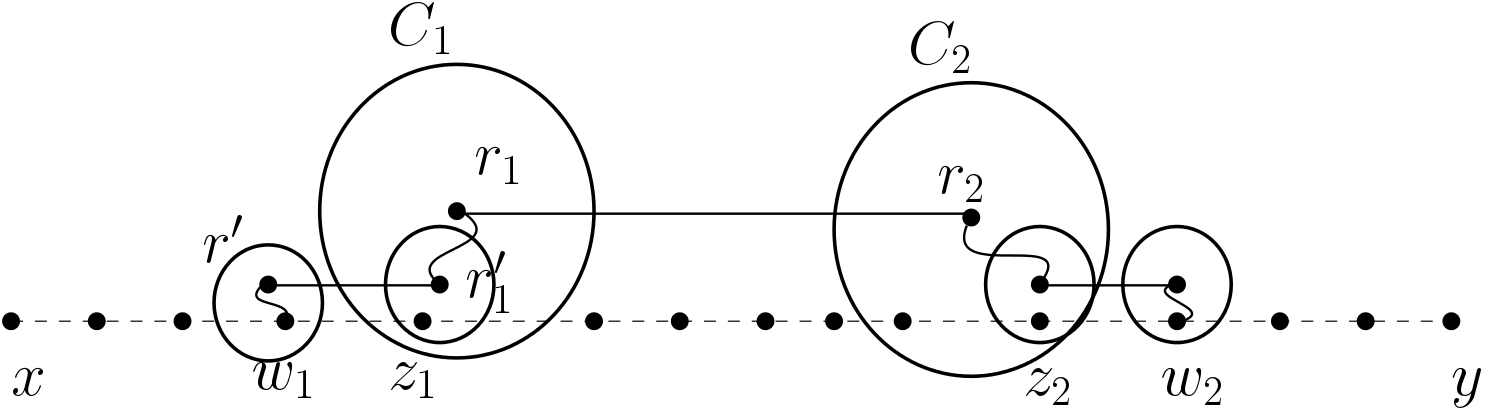}
		\caption{The path in $H$ between $w_1,r_1$, where $z_1$ is the first $U_i$-clustered vertex on the segment $S$, $w_1$ is the predecessor of $z_1$ on $S$, and $r_1$ is the center of the cluster $C_1$ such that $z_1\in C_1$. The dotted line represents the path $x,y$ in $G$. The straight solid lines represent edges of the emulator, and the curved lines represent paths in the emulator between vertices and the centers of their respective clusters in phase $j$.} 		
		\label{fig emu path2}
		\end{minipage}%
	\end{figure}

	Next, we bound the distances $d_H(x,r_1),d_H(r_2,y)$. 	
	Let $w_1,w_2$ be the predecessor and successor of $z_1,z_2$ on $\pi(x,y)$, respectively. See Figure \ref{fig emu path2} for an illustration. Observe that both $w_1,w_2$ are $U^{(i-1)}$-clustered. 
	
	We will show that there is a path of length at most $2R_i$ from $w_1$ to $r_1$. (This is also true for $r_2,w_2$, and the proof is analogous.)
	
	Let $C'\in U_j$ be the cluster such that $w_1\in C'$. Observe that $j<i$. Let $C_1'\in P_{j}$ such that $z_1\in C_1'$ (by Lemma \ref{lemma laminar}, such a cluster exists). Denote $r',r_1'$ the centers of the clusters $C',C_1'$, respectively. 
	
	By Lemma \ref{lemma radpi leq ri}, we have 
	\begin{equation}
	\label{eq w1 r'}
	d_H(w_1,r')\leq R_j.
	\end{equation}
	
	In phase $j$, the cluster $C'$ joined $U_j$. Since there is an edge between the clusters $C',C_1'$, and since their radii are at most $R_j$, we have $d_G(r',r_1')\leq 2R_j+1 \leq \delta_j.$ 
	By Lemma \ref{lemma neighboring clusters cent}, we have 
	\begin{equation}\label{eq r' r1'}
		d_H(r',r_1')= d_G(r',r_1')\leq 2R_j+1 .
	\end{equation}
	 
	Since $r_1'$ belongs to the cluster $C_1$, we also have by Lemma \ref{lemma radpi leq ri} that
	\begin{equation}\label{eq r1' r1}
		d_H(r_1',r_1)\leq R_i.
	\end{equation}
	 
	 By \cref{eq w1 r',eq r' r1',eq r1' r1} we have 
	%\begin{equation}\label{eq w1 r1}
	$d_H(w_1,r_1)\leq 3R_j+1+R_i$.
	%\end{equation}
	Observe that $R_i = 2\eps{i-1} + 5R_{i-1} \geq 2 +5R_j$. It follows that 
	\begin{equation}	\label{eq 2ri}
	d_H(w_1,r_1)\leq 2 R_i.
	\end{equation}
	Similarly, we have
	\begin{equation}
	\label{eq 2ri2}
	d_H(r_2,w_2)\leq 2 R_i.
	\end{equation}

	Recall that all vertices on the subpaths of $\pi(x,y)$ between $x,w_1$ and $w_2,y$, are $U^{(i-1)} $-clustered. Thus the induction hypothesis is applicable to them. It follows that: 	
	\begin{equation}
	\begin{array}{clclclclclc}\label{eq induc path}
	d_H(x,w_1) &\leq & \alpha_{i-1}\cdot d_G(x,w_1) + \beta_{i-1} & and &
	d_H(w_2,y) &\leq & \alpha_{i-1}\cdot d_G(w_2,y) + \beta_{i-1}
	\end{array}
	\end{equation} 
	
	By \cref{eq 2ri,eq 2ri2,eq centers dist,eq induc path} we derive: 
	\begin{equation*}
	\begin{array}{lclclc}\label{eq xy path}
	d_H(x,y) &\leq & d_H(x,w_1)+d_H(w_1,r_1)+d_H(r_1,r_2)+d_H(r_2,w_2)+d_H(w_2,y)\\
	&\leq& d_H(x,w_1)+2R_i+2R_i+d_G(z_1,z_2)+2R_i+d_H(w_2,y)\\
	&\leq& \alpha_{i-1}\cdot d_G(x,w_1)+d_G(z_1,z_2)+ \alpha_{i-1}\cdot d_G(w_2,y)+ 2\beta_{i-1}+6R_i\\
	&\leq& \alpha_{i-1}\cdot d_G(x,y) + 2\beta_{i-1}+6R_i.
	\end{array}
	\end{equation*}

	Since we divided the path $\pi(u,v)$ into $q \leq \lceil \frac{d\epsilon^i}{1-\epsilon^i}\rceil$ such segments, we obtain that for the pair $u,v$, their distance in the emulator $H$ satisfies: 
	\begin{equation*}
	\begin{array}{lclclc}\label{eq uv path}
	d_H(u,v) &\leq & \sum_{j=1}^{q} (\alpha_{i-1}\cdot d_G(x_j,y_j)+ 2\beta_{i-1}+6R_i)\\
	&\leq & \alpha_{i-1}\cdot d_G(u,v)+ (d_G(u,v)\cdot \frac{\epsilon^i}{1-\epsilon^i}+1)\cdot ( 2\beta_{i-1}+6R_i)\\
	
	&\leq & d_G(u,v)\cdot \left( \alpha_{i-1} + \frac{\epsilon^i}{1-\epsilon^i}\cdot ( 2\beta_{i-1}+6R_i)\right) + 2\beta_{i-1}+6R_i .
	\end{array}
	\end{equation*}

	Recall that $\beta_i = 2\beta_{i-1}+6R_i$ and that $\alpha_i = \alpha_{i-1} +\frac{\epsilon^i}{1-\epsilon^i}\cdot \beta_i$. It follows that:
	\begin{equation*}
	\begin{array}{lclclc}\label{eq uv path2}
	d_H(u,v) 
	&\leq & \left( \alpha_{i-1} + \frac{\epsilon^i}{1-\epsilon^i}\cdot \beta_i\right)\cdot d_G(u,v)+ \beta_i 
	&= & \alpha_i\cdot d_G(u,v) + \beta_i .
	\end{array}
	\end{equation*} 	
\end{proof}

Recall that $U^{(\ell)}$ is a partition of $V$. 
As a corollary to Lemma \ref{lemma emu stretch} we have:

\begin{corollary}\label{coro st emu 1}
	For every pair of vertices $u,v\in V$, the distance between them in $H$ satisfies:
	$$d_H(u,v)\leq \alpha_\ell\cdot d_G(u,v)+ \beta_\ell.$$
\end{corollary}

It is left to provide an upper bound on $\alpha_{\ell},\beta_{\ell}$. 
Recall that $\beta_0 = 0,\alpha_0=1$, and for $i>1$ we have $\beta_i = 2\beta_{i-1}+6R_i$ and $\alpha_i = \alpha_{i-1} +\frac{\epsilon^i}{1-\epsilon^i}\cdot \beta_i$.

\begin{lemma}
	\label{lemma test bound bi}
	For all $i\in [0,\ell]$, we have:
	$\beta_i =\sum_{j=0}^{i} 2^{i-j}\cdot 6 R_j .$
\end{lemma}

\begin{proof}
	The proof is by induction on the index of the phase $i$. For $i=0$, since $\beta_0 = 0$ and $R_0=0$, both sides of the equation are equal to $0$. 	
	
	We assume that the claim holds for some $i\in [0,\ell-1]$, and prove that it holds for $i+1$. 
	By the induction hypothesis we obtain:
	\begin{equation*}
	\begin{array}{lclclclclclc}
	\beta_{i+1} &=& 2\beta_{i}+6R_{i+1}
	&=& 6R_{i+1}+ 2\cdot \sum_{j=0}^{i} 2^{i-j}\cdot 6 R_j
	%	&=& 6R_{i+1}+ \sum_{j=1}^{i} 2^{i+1-j}\cdot 6 R_j\\
	&=& \sum_{j=0}^{i+1} 2^{i+1-j}\cdot 6 R_j
	\end{array}
	\end{equation*}
\end{proof}

We will now provide an explicit bound on $\beta_i$. 
By \cref{eq exp ri}
for all $i\in [1,\ell]$, we have that $R_i \leq 4\cdot \eps{i-1}$. Recall also that $R_0 =0$. Since we assume $\epsilon \leq 1/10$, we have 
\begin{equation}\label{eq explicit bi test}
\begin{array}{lclclclclclc}
\beta_i 
&\leq& 
\sum_{j=1}^{i} 2^{i-j}\cdot 6 R_j

&\leq&
\sum_{j=0}^{i} 2^{i-j}\cdot 24\cdot \eps{j-1}

%
%&\leq&18 \cdot 
%2^i\epsilon \cdot 
%\left[
%\frac{(\frac{1}{2\epsilon})^{i+1} }{\frac{1}{2\epsilon}-1}
%\right]\\
%
%
%&=& 18 \cdot 
%2^i\epsilon \cdot 
%
%(\frac{1}{2\epsilon})^{i+1}\cdot \frac{2\epsilon}{1-2\epsilon}\\

&\leq &\frac{24}{1-2\epsilon} \cdot 
\eps{i-1} &\leq & 30 \eps{i-1}.
\end{array}
\end{equation}

% ================ alpha ================

For all $i\in [1,\ell]$, we have $\alpha_{i} = \alpha_{i-1} +\frac{\epsilon^{i}}{1-\epsilon^i}\cdot \beta_{i}$. Thus, $\alpha_{i} \leq \alpha_{i-1} +\frac{30\epsilon}{1-\epsilon^i} \leq \alpha_{i-1} +34\epsilon $. (Note that $\epsilon \leq 1/10$.) Since $\alpha_0 = 1$, we have:
\begin{equation}
\label{eq alphai}
\alpha_i = 1+ 34\epsilon\cdot i .
\end{equation}

As a corollary to Corollary \ref{coro st emu 1} and \cref{eq explicit bi test,eq alphai}, we have:
\begin{corollary}
	\label{coro st emu}
	For every pair of vertices $u,v\in V$ the distance between them in the emulator $H$ satisfies: 
	$$d_H(u,v)\leq (1+ 34\epsilon\cdot \ell)\cdot d_G(u,v)+ 30 \cdot \eps{\ell-1}.$$ 
\end{corollary}

\subsubsection{Analysis of the Running Time}\label{sec cent rt}

The algorithm runs for $\ell+1$ phases. Each phase $i\in [0,\ell]$ consists of executing at most $|P_i|\leq n$ Dijkstra explorations, each requires $O({|E| + n{\log n}})$ time. By
\ref{lemma bound pi} we have that $|P_i|\leq n^{1-\frac{2^i-1}{\kappa}}$ for all $i\in [0,\ell]$. 
Recall that $\ell = \centell$. Hence, the running time of the entire algorithm is bounded by 
\begin{equation}
\label{eq rt cent}
\sum_{i=0}^{\ell} O({|E|+n {\log n}})\cdot |P_i| \leq 
 O({|E|+n {\log n}})\cdot\sum_{i=0}^{\ell} n^{1-\frac{2^i-1}{\kappa}}%\leq 
%O\left( n\cdot |E|\cdot n^2{\log n}\right)
\end{equation}

\subsubsection{Rescaling}\label{sec rescale}
Define $\epsilon' = 34\epsilon\cdot \ell$. Observe that we have $\epsilon = \frac{\epsilon'}{34\ell}$. We replace the condition $\epsilon<1/10$ with the much stronger condition $\epsilon' <1$.

Recall that $\ell = \centell$. Note that $\centell\leq {\log \kappa}$ for all $\kappa\geq 2$. 
The additive term $\beta_\ell$ now translates to: 
\begin{equation*}
\beta_\ell \leq 30 \cdot \eps{\ell-1} = 30 \cdot \left(\frac{1}{\left(\frac{\epsilon'}{34\ell}\right)}\right)^{\ell-1} 
=30 \cdot \left(\frac{34 { {\log \kappa}}}{\epsilon'}\right)^{{\log \kappa} -1} 
\end{equation*}

Denote now $\epsilon= \epsilon'$. 
\begin{corollary}\label{coro emu}
	For any parameters $\epsilon <1$ and $\kappa\geq 2$, and any $n$-vertex graph $G=(V,E)$, our algorithm constructs a $\left(1+\epsilon,\beta \right)$-emulator with at most 
	$ \nfrac $
	edges in polynomial deterministic time in the centralized model, where 
	$$\beta = \betaemu.$$
\end{corollary}

Note that be setting $\kappa= f(n)\cdot ({\log n})$, for a function $f(n)= \omega(1)$, we obtain an emulator of size at most 
$$n^{1+\frac{1}{f(n){\log n}}} = n\cdot 2^{\frac{1}{f(n)}} = n\left(1+O\left( \frac{1}{f(n)}\right) \right) = n+o(n). $$
By Corollary \ref{coro emu}, we derive: 
\begin{corollary}\label{coro emu us}
	For any parameter $\epsilon <1$ and any $n$-vertex graph $G=(V,E)$, our algorithm constructs a $\left(1+\epsilon,\beta \right)$-emulator with
	$ n+o(n) $
	edges in $poly(n)$ deterministic time in the centralized model, where $$\beta =\left(\frac{{\log {\log n}}}{\epsilon}\right)^{(1 + o(1)){\log {\log n}}}.$$
\end{corollary}

%
%
% 			NEW SECTION	
%	
%

\newcommand{\betadist}{\left(\frac{ {\log \kappa\rho}+ \rho^{-1} }{\epsilon\rho }\right)^{ {\log \kappa\rho}+ \rho^{-1}+O(1) }}

Using techniques discussed in Section \ref{sec congest} one can improve the running time in this result to $O(|E|\cdot n^\rho)$ , for an arbitrarily small parameter $\rho > 0$, at the expense of increasing $\beta$ to 

\begin{equation*}
\left(\frac{ {\log (\rho{\log n})}+ \rho^{-1} }{\epsilon\rho }\right)^{ {\log (\rho{\log n})}+ \rho^{-1}+O(1) }.
\end{equation*}
\section{A Construction of Ultra-Sparse Near-Additive Emulators in the \congest\ Model}\label{sec congest}
In this section we provide an implementation of the algorithm described in Section \ref{sec cent} in the distributed \congestmo. Here we aim at a low polynomial time, i.e., $O(n^\rho)$ for an arbitrarily small constant parameter $1/\kappa <\rho< 1/2$. Recall that $\kappa$ is a parameter that controls the size of the resulting emulator. 
We will show that 
for any parameters $\kappa\geq 2$, and ${1}/{\kappa}\leq \rho<{1}/{2}$, and any $n$-vertex unweighted undirected graph $G=(V,E)$, our algorithm constructs a $\oeb$-emulator with
at most $ \nfrac $	edges,
in 
$ O\left(n^\rho\cdot \beta \right)$
deterministic time in the \congest\ model, where $\beta =\betadist$.

In particular, by setting $\kappa = \omega({\log n})$, we obtain a $(1+\epsilon,\beta)$-emulator of size $n+o(n)$, with $\beta = 	
\left(\frac{ {\log (\rho{\log n})}+ \rho^{-1} }{\epsilon\rho }\right)^{ {\log (\rho{\log n})}+ \rho^{-1}+O(1)}$, in deterministic \congestmo\ in low polynomial time.

From this point until the end of the paper, we assume that all vertices have unique IDs such that for all $v,\ v.ID\in [0, n-1]$, and all vertices know their respective IDs. Moreover, we assume that all vertices know the number of vertices $n$. In fact, our results apply even if vertices know an estimate $\tilde{n}$ for $n$, where $n\leq \tilde{n}\leq poly(n)$, and have distinct ID numbers in the range $[1,\tilde{n}]$.

%==========

\subsection{The Construction}\label{sec emu dist construction}
As in the centralized variant of the algorithm, the distributed variant also initializes $H= \emptyset$ and proceeds in phases. The input to each phase $i\in [0,\ell]$ is a collection of clusters $P_i$, a degree parameter $deg_i$ and a distance threshold parameter $\delta_i$. 
The parameters $\ell,\{\deg_i,\delta_i\ | \ i\in[0,\ell]\}$ are slightly different in the current variant of the algorithm, and are specified in Section \ref{sec param dist}. The set $P_0$ is initialized as the partition of $V$ into singleton clusters.

In the distributed model, sequentially considering clusters requires too much time. Hence, the definition of popular clusters and cluster centers is slightly different in the distributed variant of the algorithm. 

For every index $i\in [0,\ell]$, a pair of distinct clusters $C,C'\in P_i$ and their respective centers $r_C,r_{C'}$ are said to be \emph{neighboring clusters} and \emph{neighboring cluster centers} if $d_G(r_C,r_{C'})\leq \delta_i$. A cluster $C$ and its center $r_C$ are said to be \textit{popular}, if $C$ has at least $deg_i$ neighboring clusters. 

Intuitively, each phase is divided into two consecutive steps. The \emph{superclustering step} of phase $i$ begins by detecting popular clusters from $P_i$ and clusters that have a neighboring popular cluster, and continues by grouping them into superclusters. When this step terminates, all clusters that have not been superclustered are not popular, and also, all of their neighboring clusters are not popular. Denote by $U_i$ the set of clusters from $P_i$ that did not join a supercluster during phase $i$. In the \emph{interconnection step}, clusters from $U_i$ are interconnected with their neighboring clusters. The details of the implementation of the superclustering step and the interconnection step are given in Sections \ref{sec super} and \ref{sec intercon}, respectively.

As in the centralized version, we will show that in the last phase $\ell$, we will have that $|P_\ell|\leq deg_\ell$, and therefore there are no popular clusters. Hence, the superclustering step of this phase is skipped, and we move directly to the interconnection step. 

\subsubsection{Setting Parameters}\label{sec param dist}
Define recursively $R_0= 0$, and for every $i\in [1,\ell]$ define $R_{i+1} = (\frac{4}{\rho}+2)\delta_i+R_i$. The distance threshold parameter is defined to be $\delta_i = \epsi+2R_i$, for every $i\in [0,\ell]$. 

In our distributed implementation of the algorithm, the execution of each phase $i$ requires $\Omega(deg_i)$ time. Recall that we aim at a low polynomial time. Therefore, we ensure that $deg_i\leq n^\rho$ for all phases $i\in [0,\ell]$. We divide the phases into two stages. In the exponential growth stage, that consists of phases $\{0,1,\dots,i_0 = \lfloor{\log \kappa\rho }\rfloor \}$, we set $deg_i = \degi$. For the fixed growth stage, that consists of phases $\{ i_0+1,i_0+2\dots, \ell = i_0+\frac{\kappa+1 }{\kappa\rho}-1 \}$, we set $deg_i = n^\rho$.

\subsubsection{Superclustering Step}\label{sec super}
In this section, we provide the execution details for the superclustering step of phase $i\in [0,\ell-1]$. During this step, we complete three tasks. The first task is to detect the set of popular clusters. The second task is to select representatives around which superclusters will be constructed. The third and most complicated task is to construct the superclusters around the selected representatives, such that all popular clusters are superclustered.

\textbf{Task 1: Detecting popular clusters.} 
To detect popular clusters, we employ the modified Bellman-Ford exploration, devised in \cite{ElkinMatar}. 

Generally speaking, we initiate a modified parallel Bellman-Ford exploration from the set of centers of clusters in $P_i$. 
The exploration consists of $\delta_i+1$ strides. In stride $0$, 
each vertex $v\in V$ initializes a list $\mathcal{L}(v)=\emptyset$. Each center $r_C$ of a cluster $C\in P_i$ writes the element $\langle r_C ,0\rangle$ to its list $\mathcal{L}(r_C)$. 
In every stride $j\in [1,\delta_i]$, each vertex $v\in V$ delivers to its neighbors in $G$ messages regarding the (up to) $deg_i+1$ cluster centers it has learnt about during stride $j-1$.
If a vertex has received messages regarding more than $deg_i+1$ centers during some stride $j\in [0,\delta_i-1]$, it arbitrarily chooses $deg_i+1$ of these messages to forward during stride $j+1$. 
Observe that stride $0$ requires $O(1)$ time, and each one of the strides $j\in [1,\delta_i]$ require $O(deg_i)$ communication rounds. 
When the exploration terminates, each center $r_C$ of a cluster $C\in P_i$ that has received messages regarding at least $deg_i$ other cluster centers is defined popular. Denote by $W_i$ the set of popular cluster centers. 

For completeness, the pseudo-code of the algorithm appears below. Theorem \ref{theorem popular} summarizes the properties of the algorithm. For its proof, see Theorem 2.1 in \cite{ElkinMatar}.

\begin{algorithm}[H]
	\caption{Detecting Popular Clusters}
	\label{Alg number of near neighbors}
	\begin{algorithmic}[1]
		\State \textbf{Input:} graph $G=(V,E)$, a set of clusters ${P}_i$, parameters $deg_i,\delta_i$
		\State \textbf{Output:} a set $W_i$.
		\State Each vertex $v\in V$ initializes a list $\mathcal{L}(v)= \emptyset$.
		\State Each $r_C\in S_i$ adds $\langle r_C.ID,0\rangle$ to $\mathcal{L}(r_C)$.
		\For {$j=1 \ to \ \delta_i$ }
		\For{$deg_i$ rounds}
		\If {$v$ received at most $deg_i+1$ messages $\langle r_C,j-1\rangle$}
		\State For each received messages $\langle r_C,j-1\rangle$, $v$ sends $\langle r_C,j\rangle$
		\EndIf
		\If {$v$ received more than $deg_i+1$ messages $\langle r_C,j-1\rangle$}
		\State {For arbitrary $deg_i+1$ received messages $\langle r_C,j-1\rangle$, $v$ sends $\langle r_C,j\rangle$}
		\EndIf
		\EndFor
		\EndFor 
		\State Each $r_C\in S_i$ that has learned about at least $deg_i$ other cluster centers joins $W_i$. 
	\end{algorithmic}
\end{algorithm}

\begin{theorem}\label{theorem popular} % \cite{ElkinMatar}
	Given a graph $G=(V,E)$, a collection of clusters $P_i$ centered around cluster centers $S_i$ and parameters $\delta_i,\ \deg_i$, Algorithm \ref{Alg number of near neighbors} returns a set $W_i$ in $O(deg_i\cdot\delta_i)$ time such that: 
	\begin{enumerate}
		\item $W_i$ is the set of all centers of popular clusters from $P_i$.
		\item Every cluster center $r_C\in S_i$ that did not join $W_i$ knows the identities of all the centers $r_{C'}\in S_i$ such that $d_G(r_C,r_{C'})\leq \delta_i$. Furthermore, for each pair of such centers $r_C,r_{C'}$, there is a shortest path $\pi$ between them such that all vertices on $\pi$ know their distance from $r_{C'}$.
\end{enumerate} 
\end{theorem}

%new paragraph
Alternatively\footnote{We are grateful to an anonymous reviewer of PODC'21 for pointing this to us.}, one can accomplish the task of Algorithm \ref{Alg number of near neighbors} even faster, in time $O(deg_i+\delta_i)$, via the $(S,d,k)$-source detection algorithm of Lenzen and Peleg \cite{LenzenP13}.
In the $(S,d,k)$-source detection problem, one is given a subset $S$ of sources, and two integers $d$ and $k$. The algorithm of \cite{LenzenP13} computes for every vertex $v\in V$ at most $k$ sources $s\in S$ that satisfy $d_G(v,s)\leq d$. The running time of their (deterministic) algorithm is $O({\min\{ d,D\}  }+{\min\{ k,|S|\}  })$. In our case, $S= S_i$, $d=\delta_i$, $k= deg_i$, and as a result the running time is $O(deg_i+\delta_i)$. The algorithm of \cite{LenzenP13} can also produce the shortest path $\pi$ between cluster centers $r_C$ and   $r_{C'}$ as above, within the same running time. 
Our algorithm, however, has a number of other steps that require $O(deg_i\cdot \delta_i)$ time, and thus using the (simpler) algorithm given in Algorithm \ref{Alg number of near neighbors} for detecting popular clusters is good enough for our purposes. 
%=======
%=======
%=======
%=======
%=======
%=======

\textbf{Task 2: Selecting representatives.}
To select a subset of the popular clusters, we compute a $(2\delta_i+1,2\delta_i/\rho)$-ruling set for $W_i$ w.r.t. the graph $G$. See Section \ref{sec def ruling setes} for the definition of ruling sets. 
This is done using the algorithm of \cite{sew,KuhnMW18}. Theorem \ref{theorem ruling set} summarizes the properties of the returned ruling set $S_i$. 

\begin{theorem}
	\label{theorem ruling set}
	{\normalfont \cite{sew,KuhnMW18}}
	Given a graph $G=(V,E)$, a set of vertices $W_i\subseteq V$ and parameters $q\in\{1,2,\ldots \},c>1$, 
	one can compute a $(q+1,cq)$-ruling subset for $W_i$ in $O(q\cdot c\cdot n^\frac{1}{c})$ deterministic time, in the \congestmo. 
\end{theorem}

For the sake of brevity, denote $sep_i =2\delta_i+1 $ and $rul_i =({2}/{\rho})\cdot\delta_i $.
By Theorem \ref{theorem ruling set}, the returned subset $S_i$ is a $(sep_i,rul_i)$\textit{-ruling set} for the set of popular clusters $W_i$.

%=======
%=======
%=======
%=======
%=======
%=======
\textbf{Task 3: Constructing superclusters.}
First, a BFS exploration rooted at the ruling set $S_i$ is executed to depth $rul_i+\delta_i$ in $G$. As a result, a forest $F_i$ is constructed, rooted at vertices of $S_i$. 

Consider a cluster center $r_C\in S_i$, and let $T_C$ be its tree in the forest $F_i$.
A cluster $C'$ is said to be spanned by $T_C$ if its center $r_{C'}$ is spanned by $T_C$. Intuitively, we would like to form a new supercluster $\widehat{C}$ centered around $r_C$, that will contain all the clusters $C'$ spanned by $T_C$. This requires informing $r_C$ of all the centers of clusters that are spanned by $T_C$, which may cause significant congestion. Therefore, we use a different approach, that may form several superclusters that will cover all clusters spanned by $T_C$.

%==========

To form superclusters, we backtrack the BFS exploration that has created $T_C$. The backtracking procedure operates for $rul_i+\delta_i$ strides, each consists of $\lfloor2deg_i\rfloor+2$ communication rounds. During each stride $d$, each vertex $v\in V$ that is spanned by $T_C$ and has $d_{T_C}(r_C,v) = rul_i+\delta_i-d$ sends messages to its parent in the tree $T_C$. For $d=0$, let $M=\emptyset$. For $d>0$, let $M$ be the set of messages that $v$ has received during stride $d-1$ of the procedure. If $v$ is a center of a cluster from $P_i$, it adds the 
message $m_v = \langle v, d_G(r_C,v) \rangle$ to $M$. If the number of messages in $M$ is smaller than $2deg_i+2$, then $v$ sends all the messages in $M$ to its parent w.r.t. $T_C$ during stride $d$.

Consider the case where $|M|\geq 2deg_i+2$. In this case, we say that $v$ is a \emph{hub-vertex}. Since it cannot send $|M|$ messages to its parent in $T_C$, the vertex $v$ decides to split from $T_C$ and form new superclusters. 
Note that the vertex $v$ receives messages from its children in the tree $T_C$ only during stride $d-1$, where $d= d_{T_C}(r_C,v)$. 

If $v$ is a center of a cluster from $P_i$, it forms a single new supercluster $\widehat{C}_v$, and $v$ is set to be the center of the new supercluster $\widehat{C}_v$. For every message $m_{r_{C'}}=\langle r_{C'}, d_G(r_C,r_{C'})\rangle$ in $M$, it adds the edge $(v,r_{C'})$ to the emulator with weight $d_G(v,r_{C'}) = d_G(r_C,r_{C'})- d_G(r_C,v)$. The vertex $v$ informs $r_{C'}$ of the new edge and its weight. This is done by sending the message $m_{r_{C'}}^r=\langle (v,r_{C'}), d_G(v,r_{C'})\rangle$ along the same route that the message $m_{r_{C'}}$ has traversed.

If $v$ is not a center of a supercluster from $P_i$, we do not allow it to be a center of a cluster of $P_{i+1}$, and therefore it forms other superclusters. The vertex $v$ partitions its children in $T_C$ into sets $V_1,V_2,\dots, V_t$, such that for every $j\in [1,t]$ the number of messages that $v$ has received from all vertices in $V_j$ is between $2deg_i+2$ and $6deg_i+6$. 
For every $j\in [1,t]$, let $Z_j$ be the set of vertices in $T_C$ that have sent messages that have arrived $v$ via a vertex in $V_j$.
Intuitively, a supercluster $\widehat{C}_j$ is formed for every $j\in [1,t]$. This supercluster will contain every cluster $C'\in P_i$ such that its centers $r_{C'}$ is in $Z_j$. See Figure \ref{fig dist sup} for an illustration.

Partitioning the children of $v$ into $t$ sets is done in the following way. For a set $X$ of children of $v$, denote by $M(X)$ the set of messages that $v$ has received from all vertices in $X$. 
The vertex $v$ greedily adds its children into sets $V_1,V_2,\dots, V_{t'}$, such that each set $V_j$ is filled until $|M(V_j)|$ is at most $4deg_i+4$. Note that since the number of messages received by $v$ from each one of its children is less than $2deg_i+2$, we have that $|M(V_j)|\geq 2deg_i+2$, for any $j\in [1,t'-1]$. If for the last set $V_{t'}$ we have $|M(V_{t'})|< 2deg_i+2$, we add the set $V_{t'}$ to the set $V_{t'-1}$. Let $t$ be the number of sets formed by this process (i.e., if $|M(V_{t'})|< 2deg_i+2$ then $t= t'-1$. Otherwise, $t=t'$). Now we have that $2deg_i+2 \leq |M(V_j)| \leq 6deg_i+6$ for every $j\in [1,t]$.

For every $j\in [1,t]$, the vertex $v$ selects a single vertex $r\in Z_j$ to be the center of $\widehat{C}_j$. Then, $v$ must inform all vertices in $Z_j$ that their attempt to join $\widehat{C}$ has failed, and provide information regarding their new cluster center and superclustering edge. 
To this aim, we define the tree $T_C^j$ to be the tree that contains all paths from $T_C$ between a vertex in $Z_j$ and $v$. Observe that $T_C^j$ does not contain any other hub-vertices.

The vertex $v$ broadcasts the message $\langle r\rangle$ in the tree $T_C^j$. This informs all centers in $Z_j$ that their attempt to join $\widehat{C}$ has failed, and that the center of their new supercluster is $r$. In addition, for every $r'\in Z_j$, the vertex $v$ broadcasts the message $\langle r', d_G(r',v)+d_G(v,r)\rangle$ to all vertices in $T_C^j$. 
In particular, this step informs the vertices $r,r'$ that the edge $(r,r')$
was added to the emulator $H$ with weight $d_G(r',v)+d_G(v,r)$. 
Observe that the vertex $v$ knows $d_G(r_C,v)$, $d_G(r_C,r)$ and $d_G(r_C,r')$, and since it belongs to shortest $r_C-r$ and $r_C-r'$ paths, it can infer $d_G(v,r)$ and $d_G(v,r')$. In addition, 
by triangle inequality, these superclustering edges never shorten distances w.r.t. the graph $G$. 

After the $rul_i+\delta_i$ strides terminate, for every message $\langle r_{C'}, d_G(r_C,r_{C'}) \rangle$ that arrives to $r_C$, the edge $(r_C,r_{C'})$ is added to the emulator $H$ with weight $d_G(r_C,r_{C'})$. All vertices that belong to the cluster $C'$ centered around $r_{C'}$ join the supercluster $\widehat{C}$. 
This completes the description of the procedure for forming superclusters. 

\begin{figure}
	\centering 
	\includegraphics[scale=0.1]{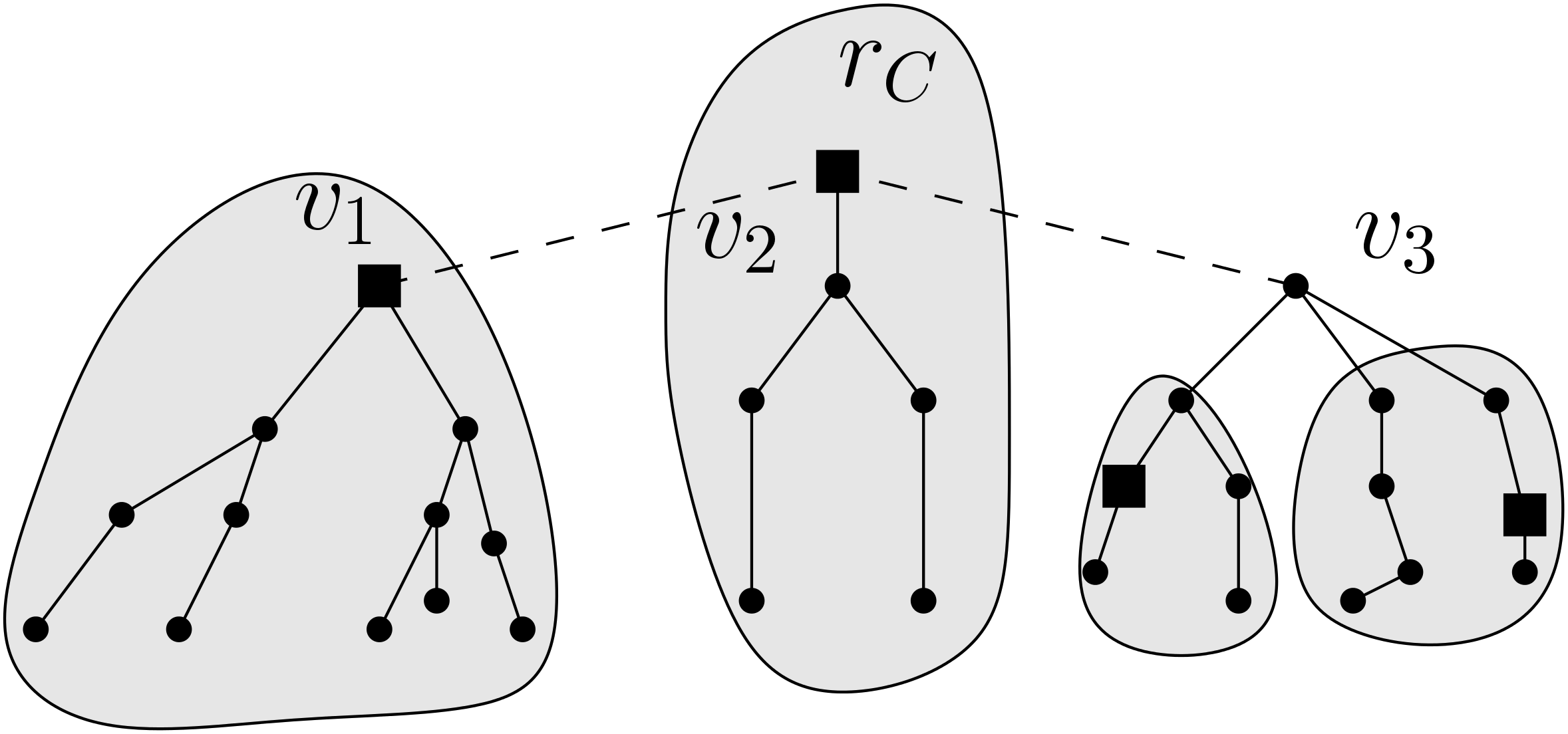}
	\caption{Forming superclusters that cover all clusters of $P_i$ spanned by the BFS tree $T_C$. In the figure, the vertex $r_C$ is the root of the tree $T_C$. The four gray areas are four superclusters, and the big squares are their respective centers. The dashed lines from $v_1,v_3$ to $r_C$ represent the fact that $v_1,v_3$ have more than $2deg_i+2$ messages to deliver to $r_C$. The vertex $v_1$ is a center of a cluster of $P_i$, and so it becomes the center of a new supercluster. The vertex $v_3$ is not a center of a supercluster, and so it divides its children into two sets, and chooses a representative from each set to become the center of its respective supercluster.}
		\label{fig dist sup}
\end{figure}

In the following lemma, we analyze the running time of Task 3, and show that all computation terminate within the $rul_i+\delta_i$ strides. 

\begin{lemma}\label{lemma task 3 rt} 
	For every index $i\in [0,\ell-1]$, Task 3 requires $O(\frac{\delta_i}{\rho}\cdot deg_i)$ communication rounds. 
\end{lemma}
\begin{proof}
	Consider an index $i\in [0,\ell-1]$. On our way to constructing superclusters, we execute a BFS exploration that forms the forest $F_i$ and a backtracking procedure. In addition, we take care of \emph{hub-vertices}.

	The BFS exploration that forms the forest $F_i$ is executed to depth $rul_i+\delta_i$ from the set of vertices in $S_i$. Thus, it requires $O(rul_i+\delta_i)$ time. The backtracking procedure of phase $i$ is done in $rul_i+\delta_i$ strides, where each stride consists of $O(deg_i)$ communication rounds. Thus, its overall running time is $O((rul_i+\delta_i)deg_i)$.
	
	It is left to analyze the time required to take care of \emph{hub-vertices}. 	
	Let $v$ be a hub-vertex in some tree $T_C\in F_i$. 
	If $v$ is a center of a cluster of $P_i$, it forms a new supercluster around itself. 
	For every vertex $r_{C'}$ such that $v$ has received the message $m_{r_{C'}}$ from $r_{C'}$, the vertex $v$ sends a message $m_{r_{C'}}^r$. This message is sent along the same route that the message $m_{r_{C'}}$ has traversed. Note that $v$ has received less than $2deg_i+2$ messages from each one of its children. Therefore, $v$ sends less than $2deg_i+2$ messages to each one of its children. The distance each such message is required to traverse is at most $rul_i+\delta_i$. Hence, all messages returned by $v$ to their senders arrive within $O(rul_i+\delta_i+deg_i)$ communication rounds (via pipelined broadcast).
	
	If $v$ is not a center of a cluster in $P_i$, its partitions its children into sets $V_1,V_2,\dots, V_{t}$. Note that this partition is computed locally.
	Recall that for every $j\in [1,t]$, the number of messages that $v$ has received from all vertices in $V_j$ is $|Z_j|\leq 6deg_i+6$. Also recall that $v$ broadcasts $O(|Z_j|)$ messages along the edges of the tree $T_C^j$. The depth of the tree is at most $O(rul_i+\delta_i)$.
	Hence, a broadcast that originated from $v$ terminates within $O(rul_i+\delta_i+deg_i)$ communication rounds (via pipelined broadcast).
	
	Recall that $rul_i = ({2}/{\rho})\cdot\delta_i$. 
	It follows that the overall time required to complete Task 3 in phase $i$ is 
	$O((rul_i+\delta_i)deg_i) =O(\frac{\delta_i}{\rho}\cdot deg_i).$
\end{proof}

Next, we prove that all popular clusters in $P_i$ and clusters in $P_i$ that have a popular neighboring cluster are superclustered into superclusters of $P_{i+1}$. 

\begin{lemma}\label{lemma popular are clustered}
	Consider a cluster $C\in P_i$. If $C$ is popular or it has a neighboring cluster that is popular, then $C$ belongs to a supercluster of $P_{i+1}$.
\end{lemma}
\begin{proof}
	Recall that $W_i$ is the set of centers of popular clusters from $P_i$ and that $S_i$ is a $(sep_i,rul_i)$-ruling set for $W_i$.
	Consider a cluster $C\in P_i$. If $C$ is popular, then by definition of $S_i$ there exists a vertex $v\in S_i$ such that the distance between $v$ and the center of $C$ is at most $rul_i$. Hence, the cluster $C$ became superclustered during the superclustering step of phase $i$. 
	
	If $C$ is not popular but it has a neighboring cluster $C'$ that is popular, then there exists a vertex $v\in S_i$ such that the distance between $v$ and the center of $C'$ is at most $rul_i$. The distance between the centers of $C,C'$ is at most $\delta_i$. Therefore, there exists a vertex $v\in S_i$ with distance at most $rul_i+\delta_i$ from the center of $C$. Hence, the cluster $C$ became superclustered during the superclustering step of phase $i$. 
\end{proof}

\subsubsection{Interconnection Step}\label{sec intercon}
In this section, we provide the execution details for the interconnection step of phase $i\in [0,\ell]$. 
Denote by $U_i$ the set of clusters from $P_i$ that have not been superclustered during the superclustering step of phase $i$. 
In the interconnection step of phase $i$, each cluster in $U_i$ is connected with all its neighboring clusters from $P_i$. 

For $i\in [0,\ell-1]$, by Lemma \ref{lemma popular are clustered} we have that every cluster $C\in U_i$ is not popular. By Theorem \ref{theorem popular} the center $r_C$ of the cluster $C$ already knows the identities and distances to all its neighboring cluster centers, since Algorithm \ref{Alg number of near neighbors} was executed during the superclustering step. 
Hence, the center $r_C$ knows which edges it needs to add to the emulator $H$, as well as their weights. 

Let $r_{C'}$ be a neighboring cluster center of $r_C$. The center $r_C$ must inform $r_{C'}$ that $C$ has joined $U_i$, and therefore the edge $(r_C,r_{C'})$ is added to the emulator $H$. For this aim, we employ Algorithm \ref{Alg number of near neighbors}, as in the superclustering step of phase $i$, from the centers of all clusters in $U_i$. Recall that by Lemma \ref{lemma popular are clustered} we have that every cluster that has a neighboring cluster that is popular is superclustered. Since $C\in U_i$, we conclude that $r_{C'}$ is not popular. Hence, by Theorem \ref{theorem popular} we have that during the current execution of Algorithm \ref{Alg number of near neighbors}, the cluster $r_{C'}$ has received a message from $r_C$, and thus it knows the identity of $r_C$ and the distance $d_G(r_C,r_{C'})$. Now, both endpoints of the edge $(r_C,r_{C'})$ know that it is added to the emulator $H$ with weight $d_G(r_C,r_{C'})$. 

The interconnection step of phase $\ell$ is slightly different.
Recall that the superclustering step of phase $\ell$ is skipped. However, as we later show (\cref{eq explicit pl}), all clusters in $P_\ell$ are not popular. In the interconnection step of phase $\ell$, we execute Algorithm \ref{Alg number of near neighbors} from the set $P_\ell$ with parameters $deg_\ell,\delta_\ell$. By Theorem \ref{theorem popular}, we are guaranteed that all centers of clusters in $P_\ell$ know the identities and distance to all their neighboring cluster centers. Therefore, each center of a cluster $C\in P_\ell$ can add to the emulator $H$ edges to all its neighboring cluster centers. 
This completes the description of the interconnection step of phase $i$.

\subsection{Analysis of the Construction}
In this section, we analyze the size and stretch of the emulator, and the running time that is required to construct it.
We begin by showing that in the final phase $\ell$, we have $|P_\ell|\leq deg_\ell$. Thus, there are no popular clusters, and the superclustering step can be safely skipped. To do so, we show that the number of clusters in $P_{i+1}$ is significantly smaller than the number of clusters in $P_i$. 

Recall that during the superclustering step of every phase $i\in [0,\ell-1]$, a BFS exploration is executed from the ruling set $S_i$, and that $F_i$ is the ruling forest produced by this exploration. 

\begin{lemma}
	\label{lemma bouns pi1}
	Let $i\in [0,\ell-1]$. For every tree $T_C$ in $F_i$, let $x_C$ be the number of clusters of $P_i$ that are spanned by $T_C$. Then, the number of superclusters formed around vertices of $T_C$ is at most $x_C/(deg_i+1)$.
\end{lemma}

\begin{proof}
	Consider a tree $T_C$ in $F_i$, and let $r_C\in S_i$ be the root of $T_C$. Recall that $r_C$ is the center of a popular cluster $r_C$. 
	Let $s$ be the number of superclusters formed around vertices of $T_C$.

	If $s = 1$, then we have that all clusters that are spanned by $T_C$ belong to the supercluster formed around the center $r_C$. Recall that the set $S_i$ is a $(sep_i,rul_i)$-ruling set, where $sep_i = 2\delta_i+1$. 
	Hence, for every neighboring cluster center $r_{C'}$ of $r_C$, we have that $d_G(r_C,r_{C'})\leq \delta_i$ and thus, $r_C$ is the closest vertex to $r_{C'}$ in $S_i$. Therefore, $r_{C'}$ is spanned by $T_C$. Since $C$ is a popular cluster, it has at least $deg_i$ neighbors. Note that $C$ itself is also spanned by $T_C$. Hence, we have that $x_C \geq deg_i+1$, and the claim holds.

	Consider the case where $s\geq 2$. Every cluster formed around a vertex $v\in T_C$, where $v\neq r_C$ contains at least $2deg_i+2$ clusters of $P_i$ that are spanned by $T_C$. 
	We conclude that 
	
	\begin{equation*}
	\label{eq numtree}
	\begin{array}{lclclclclc}
	\frac{x_C} {deg_i+1} &\geq & 
	\frac{(s-1)\cdot (2deg_i+2)} {deg_i+1} & = & 2s-2 &\geq & s.
	\end{array}
	\end{equation*}
	
\end{proof}

Observe that Lemma \ref{lemma bouns pi1} implies that for every $i\in [0,\ell-1]$ we have that 
\begin{equation}
\label{eq explicit pi bound}
|P_{i+1}| \leq |P_i|\cdot deg_i^{-1}.
\end{equation}

Next, we provide an explicit bound on the number of superclusters formed in phase $i$. 

\begin{lemma}
	\label{lemma exp pi1}
	For every $i\in [0,i_0+1]$, we have that $|P_{i}|\leq n^{1-\frac{2^i-1}{\kappa}}$.
\end{lemma}
\begin{proof}
	The proof is by induction on the index of the phase $i$. For the base case, we have that $|P_0|=0$, and also $n^{1-(2^0-1)/\kappa} = n$.
	
	Assume that the claim holds for some $i\in [0,i_0]$, and prove that it holds for $i+1$. 
	
	By the induction hypothesis and by \cref{eq explicit pi bound} we have
	
	\begin{equation*}
	\label{eq explicit pi bound2}
	|P_{i+1}| \leq |P_i|\cdot deg_i^{-1} < n^{1-\frac{2^i-1}{\kappa}}\cdot n^{-\frac{2^i}{\kappa}} = n^{1-\frac{2^{i+1}-1}{\kappa}}. 
	\end{equation*}
\end{proof}

\begin{lemma}
	\label{lemma exp pi2}
	For every $i\in [i_0+1,\ell]$, we have that $|P_{i}|\leq n^{1-\frac{2^{i_0+1}-1}{\kappa}-(i-i_0-1)\rho}$.
\end{lemma}
\begin{proof}
	The proof is by induction on the index of the phase $i$. The base case $(i=i_0+1)$ holds since by Lemma \ref{lemma exp pi1} we have
	\begin{equation}
	\label{eq bound pi0}
	|P_{i_0+1}| \leq n^{1-\frac{2^{i_0+1}-1}{\kappa}}.
	\end{equation}

	Assume the claim holds for some $i\in [i_0+1,\ell-1]$, and prove that it holds for $i+1$. 
	By the induction hypothesis and by \cref{eq explicit pi bound} we have
	
	\begin{equation*}
	\label{eq explicit pi bound3}
	|P_{i+1}| \leq |P_i|\cdot deg_i^{-1} < n^{1-\frac{2^{i_0+1}-1}{\kappa}-(i-i_0-1)\rho} \cdot n^{-\rho} = 
	n^{1-\frac{2^{i_0+1}-1}{\kappa}-(i-i_0)\rho}.
	\end{equation*}
\end{proof}

Recall that $i_0 = \ize $ and that $\ell = \distell$. 
Note that $\rho \leq 1/2$, hence $\ell > i_0$, and so $deg_\ell = n^\rho$. By Lemma \ref{lemma exp pi2}, the size of the set $P_{\ell}$ satisfies 

\begin{equation}
\label{eq explicit pl}
|P_{\ell}| \leq 
n^{1-\frac{2^{\ize+1}-1}{\kappa}-(\lceil \frac{\kappa+1}{\kappa\rho}\rceil -2)\rho}
\leq 
n^{1-\frac{\kappa\rho-1}{\kappa}-\frac{\kappa+1}{\kappa} +2\rho}
\leq 
n^\rho.
\end{equation}

As a result, in the last phase $\ell$ we have $|P_\ell|\leq n^\rho = deg_\ell$, and there are no popular clusters in $P_\ell$. 

\subsubsection{Analysis of the Number of Edges}\label{sec dist analysis of size}
In this section, we analyze the number of edges added to the emulator $H$ by the algorithm. The analysis follows the lines of the corresponding arguments in Section \ref{sec cent analysis of size}. 
As in the centralized construction, here we can also charge each edge that is added to the emulator $H$ in some phase $i$ of the algorithm to a center of a cluster in $P_i$.

Every superclustering edge added to the emulator during phase $i$ can be charged to a center of a cluster $C\in P_i$ that neither joined $U_i$ nor was it selected to grow a supercluster around it during this phase. Therefore, the number of superclustering edges added to the emulator $H$ during phase $i$ is exactly $|P_i|-|U_i|-|P_{i+1}|$. 
Every interconnection edge added to the emulator during some phase $i\in [0,\ell]$ is charged to a center of a cluster $C\in U_i$. 
Recall that for every cluster $C\in U_i$, its center $r_C$ is charged with less than $deg_i$ edges. Hence, the number of interconnection edges added to the emulator during phase $i$ of the algorithm is at most $|U_i|\cdot deg_i$. (The equality is achieved when $U_i=\emptyset$.)
In addition, Lemma \ref{lemma bouns pi1} implies that the number of clusters from $P_i$ that belong to superclusters of $P_{i+1}$ is at least $|P_{i+1}|\cdot (deg_i+1)$. Hence, the number of clusters in $P_i$ that did not join a supercluster during phase $i$ satisfies $|U_i|\leq |P_i|-|P_{i+1}|\cdot (deg_i+1)$. Thus, the number of edges added to the emulator by all phases of the algorithm is bounded by 

\begin{equation}
\label{eq edges bound dist}
\begin{array}{lclclclclclclc}

|H| & \leq & \sum_{i=0}^{\ell-1} \left( |P_i|-|U_i|-|P_{i+1}| + |U_i|\cdot deg_i \right) + |P_\ell|\cdot deg_\ell \\ 
& = & \sum_{i=0}^{\ell-1} \left( |P_i|-|P_{i+1}| + |U_i|\cdot (deg_i-1) \right) + |P_\ell|\cdot deg_\ell \\
& \leq & \sum_{i=0}^{\ell-1} \left( |P_i|-|P_{i+1}| + (|P_i|-|P_{i+1}|\cdot (deg_i+1)) (deg_i-1) \right) + |P_\ell|\cdot deg_\ell \\

& = & \sum_{i=0}^{\ell-1} \left( |P_i|\cdot deg_i -|P_{i+1}|\cdot deg^2_i \right) + |P_\ell|\cdot deg_\ell \\

& = & | P_0 |\cdot deg_0 + \sum_{i=0}^{\ell-1} \left( |P_{i+1}|\cdot deg_{i+1} -|P_{i+1}|\cdot deg^2_{i} \right) .

\end{array}
\end{equation}

We now show that for all $i\in [0,\ell-1]$, we have $deg_{i+1} \leq deg^2_i$. 
Recall that in the exponential growth stage, $deg_i = \degi$, and so $deg_{i+1} = deg^2_i$ for any $i\in [0,i_0-1]$. 
Recall that $i_0 = \ize$, and so $deg^2_{i_0} = n^\frac{2^{i_0+1}}{\kappa} \geq n^\rho = deg_{i_0+1}$.
In the fixed growth stage, $deg_i = n^\rho$, and so $deg_{i+1} <deg^2_i$ for all $i\in [i_0+1, \ell-1]$. 
 It follows that for all $i\in [0,\ell-1]$, we have $deg_{i+1} \leq deg^2_i$. Thus, by \cref{eq edges bound dist}, the size of the emulator $H$ is at most

\begin{equation}
\label{eq edges bound dist fin}
\begin{array}{lclclclclclclc}

|H| & \leq & | P_0 |\cdot deg_0 & = & n^{1+1/\kappa}.

\end{array}
\end{equation}

\subsubsection{Analysis of the Stretch}\label{sec dist analysis of stretch}

In this section we analyze the stretch of the emulator $H$. We follow the lines of the analysis given in Section \ref{sec cent stretch}.
We begin by proving that $R_i$ is an upper bound on the radii of clusters in $P_i$, for all $i\in [0,\ell]$. Recall that 
$R_0= 0$, and for every $i\in [1,\ell]$, we have $R_{i+1} = (\frac{4}{\rho}+2)\delta_i+R_i $ (see Section \ref{sec param dist}).
Also recall that for every phase $i\in [0,\ell]$, a BFS ruling forest $F_i$ was constructed by a BFS exploration that was executed to depth $rul_i+\delta_i = \frac{2\delta_i}{\rho}+\delta_i = ({2}/{\rho}+1)\cdot\delta_i $. Hence the radius of each tree in $F_i$ (i.e., the maximal distance between the root of the tree and a vertex spanned by the tree) is at most $(\frac{2}{\rho} +1)\delta_i$ as well. 
For each tree $T\in F_i$, let $Rad(T)$ denote its radius, and $Rad(F_i) = \max_{T\in F_i} Rad(T)$. 
% rul_i = 2/rho \cdot \delta_i, and we also have an addition \delta_i because we need to cover all those which have neighbors that are popular. 

%$\delta_i = \epsi+2R_i$, for every $i\in [0,\ell]$. Also recall that 
%$R_0= 0$, and for every $i\in [1,\ell]$, we have $R_{i+1} = \frac{4}{\rho}\delta_i+R_i $.

\begin{lemma}\label{lemma radpi leq ri dist}
	For every index $i\in [0,\ell]$, we have $Rad(P_i)\leq R_i$.
\end{lemma}
\begin{proof}
	The proof is by induction on the index of the phase $i$. For $i=0$, all clusters in $P_i$ are singletons, and also $R_0=0$. Thus the claim holds. 
	
	Assume the claim holds for $i\in [0,\ell-1]$ and prove that it holds for $i+1$. 
	Consider a cluster $\widehat{C}\in P_{i+1}$. This cluster was formed around a vertex $r_C$ during phase $i$. Let $C\in P_i$ be the cluster rooted at $r_C$. Consider a vertex $u\in \widehat{C}$. 
	
	\textbf{Case 1:} The vertex $u$ belongs to the cluster $C$. Then, by the induction hypothesis, we have $d_H(r_C,u)\leq R_i\leq R_{i+1}$.

	\textbf{Case 2:} The vertex $u$ belonged to a cluster $C'\in P_i$, where $C\neq C'$. Denote by $r_{C'}$ the center of the cluster $C'$. The centers $r_C$ and $r_{C'}$ are both spanned by the same tree $T$ in $F_i$. Therefore, we have $d_T(r_C,r_{C'}) \leq 2Rad(T) \leq(\frac{4}{\rho}+2)\delta_i$.	
	When $r_{C'}$ joined the supercluster $\widehat{C}$, the edge $(r_C,r_{C'})$ was added to the emulator $H$, with weight $d_{T}(r_C,r_{C'}) \leq (\frac{4}{\rho}+2)\delta_i$. 

	 In addition, by the induction hypothesis, we have $d_H(u,r_{C'}) \leq R_{i}$. Hence, $$d_H(r_C,u) \leq (\frac{4}{\rho}+2)\delta_i + R_i = R_{i+1}.$$
	
\end{proof}

We now provide an explicit upper bound on $R_i$. Recall that for every $i\in [1,\ell]$, we have $\delta_i = \epsi+2R_i$, and therefore
 $$R_{i+1} = \left(\frac{4}{\rho}+2\right)\delta_i+R_i = \left(\frac{4}{\rho}+2\right)\epsi+\left(\frac{8}{\rho}+5\right)R_i.$$

\begin{lemma}
	\label{lemma bound ri dist}
	For every index $i\in [0,\ell]$, we have 
	$$R_i =\left( {\frac{4}{\rho}} +2\right)\cdot\sum_{j=0}^{i-1} \eps{j}\cdot \left(\frac{8}{\rho}+5\right)^{i-1-j}.$$
\end{lemma}
\begin{proof}
	The proof is by induction on the index $i$. For $i=0$, both sides of the equation are equal to $0$, and so the base case holds. 
	
	Assume that the claim holds for some $i\in [0,\ell-1]$, and prove that it holds for $i+1$. 
	By definition and the induction hypothesis we have: 
	\begin{equation*}
	\begin{array}{clllll}
	R_{i+1} &=& 
	
	(\frac{4}{\rho}+2)\epsi+(\frac{8}{\rho}+5)\cdot ( {\frac{4}{\rho}} +2)\cdot\sum_{j=0}^{i-1} \eps{j}\cdot \left(\frac{8}{\rho}+5\right)^{i-1-j} \\
	
	 &=&

	 ( {\frac{4}{\rho}} +2)\cdot\sum_{j=0}^{i} \eps{j}\cdot \left(\frac{8}{\rho}+5\right)^{i-j}. \\

	\end{array}
	\end{equation*}
\end{proof} 

By Lemma \ref{lemma bound ri dist}, we derive the following explicit bound on $R_i$, for $i\in [0,\ell]$. 
\begin{equation*}
\begin{array}{lclclclclclc}
R_i & = &
( {\frac{4}{\rho}} +2)\cdot\sum_{j=0}^{i-1} \eps{j}\cdot \left(\frac{8}{\rho}+5\right)^{i-1-j}
\\
& = &
( {\frac{4}{\rho}} +2)\cdot\left(\frac{8}{\rho}+5\right)^{i-1} \cdot

\sum_{j=0}^{i-1} \eps{j}\cdot \left(\frac{\rho}{8+5\rho}\right)^{j} 

\\
& \leq &
( {\frac{4}{\rho}} +2)\cdot\left(\frac{8}{\rho}+5\right)^{i-1} \cdot 
\left[\frac{\left(\frac{\rho}{\epsilon(8+5\rho)}\right)^{i} }{\frac{\rho}{\epsilon(8+5\rho)} -1 } 
\right]

\\
& = &
( {\frac{4}{\rho}} +2)\cdot\left(\frac{8+5\rho}{\rho}\right)^{i-1} \cdot \left(\frac{\rho}{\epsilon(8+5\rho)}\right)^{i}
\cdot \frac{\epsilon(8+5\rho)}{\rho-\epsilon(8+5\rho)} 

& = &
\frac{4+2\rho}{\rho-\epsilon(8+5\rho)}\cdot \eps{i-1}. 

\end{array}
\end{equation*}

Recall that $\rho < 1/2$, and assume that $\rho\geq 25\epsilon$. It follows that for all $i\in [0,\ell]$:
\begin{equation}
\label{eq exp ri dist}
R_i \leq 
\frac{4+2\rho}{\rho-\epsilon(8+5\rho)} \cdot \eps{i-1} 
\leq
 \frac{5}{\rho -12.5\epsilon} \cdot \eps{i-1}
\leq 
\frac{10}{\rho}\cdot \eps{i-1}.
\end{equation}

%==============
%==============
%==============
%==============
%==============

As in Section \ref{sec cent stretch}, define recursively $\beta_0 = 0,\alpha_0=1$, and for $i>1$ define $\beta_i = 2\beta_{i-1}+6R_i$ and $\alpha_i = \alpha_{i-1} +\frac{\epsilon^i}{1-\epsilon^i}\cdot \beta_i$.

Recall that when a cluster is added to $U_i$, for some $i\in [0,\ell]$, the algorithm adds edges from its center to the centers of all its neighboring clusters. The weight of each such edge is set to be the distance in $G$ between its two endpoints. 
Therefore, the assertion of Lemma \ref{lemma neighboring clusters cent} also holds for the distributed construction. Thus, for every $i\in [0,\ell]$ and a center $r_C$ of a cluster $C\in U_i$, an for every neighboring cluster center $r_{C'}$ of $r_C$, we have $d_H(r_C,r_{C'})=d_G(r_C,r_{C'})$. As a result, Lemma \ref{lemma emu stretch} also holds for the distributed construction. In other words, for every pair of vertices $u,v\in V$ such that all vertices on a shortest $u-v$ path are $U^{(i)}$ clustered, we have that 
\begin{equation}
\label{eq dist uv}
d_H(u,v) \leq \alpha_i\cdot d_G(u,v)+\beta_i.
\end{equation}

Recall that $U_{-1} = \emptyset$, and $U^{(i)} = \bigcup_{j=-1}^i U_i$ for all $i\in [-1,\ell]$.
As in the centralized construction, here we also have that the set $U^{(\ell)}$ is a partition of $V$. Hence, \cref{eq dist uv} implies that for 
every pair of vertices $u,v\in V$ we have 
\begin{equation}
\label{eq dist uv ell}
d_H(u,v) \leq \alpha_\ell\cdot d_G(u,v)+\beta_\ell.
\end{equation}

By Lemma \ref{lemma test bound bi}, we have that the recursion $\beta_0 = 0$ and $\beta_i = 2\beta_{i-1}+6R_i$ for $i>1$ solves to 
\begin{equation}
	\beta_i =\sum_{j=0}^{i} 2^{i-j}\cdot 6 R_j .
\end{equation}

We will now provide an explicit bound on $\beta_i$. 
By \cref{eq exp ri dist}
for all $i\in [1,\ell]$, we have that $R_i \leq \frac{10}{\rho}\cdot \eps{i-1}$. Since we assume $\epsilon <1/10$, we have 
\begin{equation}\label{eq explicit bi test dist}
\begin{array}{lclclclclclc}
\beta_i 
&\leq& 
\sum_{j=0}^{i} 2^{i-j}\cdot 6 R_j

&\leq& \frac{60\cdot 2^i\epsilon}{\rho} \cdot 
\left[ \frac{ (\frac{1}{2\epsilon})^{i+1}} {\frac{1-2\epsilon}{2\epsilon} } 
\right]

&=& \frac{75}{\rho} \cdot \eps{i-1}.
\end{array}
\end{equation}

By \cref{eq explicit bi test dist} and since $\epsilon\leq 1/10$, the recursion $\alpha_0 = 0$ and $\alpha_{i} = \alpha_{i-1} +\frac{\epsilon^{i}}{1-\epsilon^i}\cdot \beta_{i}$ for $i>0$ solves to 
\begin{equation}
\label{eq alphai dist1}
\alpha_i = 1+ \frac{90\epsilon}{\rho}\cdot i.
\end{equation}

By \cref{eq dist uv ell,eq explicit bi test dist,eq alphai dist1} we derive the following corollary.

\begin{corollary}
	\label{coro st emu dist}
	For every pair of vertices $u,v\in V$ the distance between them in the emulator $H$ satisfies: 
	$$d_H(u,v)\leq \left(1+\frac{90\epsilon\cdot \ell }{\rho } \right)\cdot d_G(u,v)+ \frac{75}{\rho } \cdot \eps{\ell-1}.$$ 
\end{corollary}

\subsubsection{Analysis of the Running Time}\label{sec dist analysis of rt}
In this section, we analyze the running time of the algorithm. We begin by analyzing the running time of a single phase $i\in [0,\ell-1]$. 

\paragraph{Superclustering Step.} To detect the popular clusters, we executes Algorithm \ref{Alg number of near neighbors}. By Theorem \ref{theorem popular}, the algorithm requires $O(deg_i\cdot \delta_i)$ time. By Theorem \ref{theorem ruling set}, constructing a ruling set for the popular clusters requires $O(\delta_i\cdot \frac{1}{\rho} \cdot n^{\rho})$ time. 
By Lemma \ref{lemma task 3 rt},
Computing superclusters requires $O(\frac{\delta_i}{\rho}\cdot deg_i)$ time.
Hence, the superclustering step of phase $i$ can be executed in 
$O\left(\frac{\delta_i\cdot n^\rho}{\rho} \right)$ deterministic time in the \congestmo.

\paragraph{Interconnection Step.} The interconnection step consists of executing Algorithm \ref{Alg number of near neighbors}, as in the superclustering step. Hence, the running time of the interconnection step is dominated by the running time of the superclustering step.

For the final phase $\ell$, the superclustering step is skipped. The interconnection step of the final phase requires executing Algorithm \ref{Alg number of near neighbors}, in $O(deg_\ell\cdot \delta_\ell)$ time. Recall that $\delta_i = \epsi+2R_i$, and also that by \cref{eq exp ri dist} we have that $R_i \leq 
\frac{10}{\rho}\cdot \eps{i-1}$,
for every $i\in [0,\ell]$. In addition, recall that we assume $\rho > 25\epsilon$. Hence, for every $i\in [0,\ell]$ we have 
\begin{equation}
\label{eq bound deltai}
	\delta_i = O(\epsi)
\end{equation}
It follows that the running time of the entire algorithm is at most

%==============

\begin{equation}\label{eq dist rt}
\begin{array}{lclclclclclc}
O(deg_\ell\cdot \delta_\ell) + \sum_{i=0}^{\ell-1} O(\frac{n^\rho\cdot \delta_i}{\rho}) & = & 
 O\left(\frac{n^\rho}{\epsilon^\ell} + \frac{n^\rho}{\rho} \cdot \sum_{i=0}^{\ell-1} \epsi \right)& = & O\left(\frac{n^\rho}{\epsilon^\ell}\right).
\end{array}
\end{equation}

\subsubsection{Rescaling}\label{sec rescale dist}

Define $\epsilon' = \frac{90\epsilon\cdot \ell }{\rho }$. Observe that we have 
$\epsilon = \frac{\epsilon'\rho }{90\ell}$. We replace the condition $\epsilon<1/10$ with the much stronger condition $\epsilon' <1$. The assumption $\rho > 25\epsilon$ holds since $ \epsilon'< 1$. 

Recall that $\ell = \lfloor {\log \kappa \rho}\rfloor +\lceil { \frac{\kappa+1}{\kappa\rho}}\rceil -1$. Note that $\ell = {\log \kappa\rho}+ \rho^{-1}+O(1)$ for all $\kappa\geq 2$. 
The additive term $\beta_\ell$ now translates to: 
\begin{equation*}
\beta_\ell = \frac{75}{\rho } \cdot \eps{\ell-1} = 
\frac{75}{\rho } \cdot \left(\frac{90\ell}{\epsilon'\rho }\right)^{\ell-1} = 
\left(\frac{ {\log \kappa\rho}+ \rho^{-1} }{\epsilon'\rho }\right)^{ {\log \kappa\rho}+ \rho^{-1}+O(1) }.
\end{equation*}

Denote 
\begin{equation*}
\label{eq beta def dist}
\beta = 
\betadist.
\end{equation*}

By \cref{eq dist rt}, the running time of the algorithm is $
O\left(\frac{ n^\rho }{ \epsilon^{\ell}} \right)= 
O\left(\beta n^\rho\right).$
Denote now $\epsilon= \epsilon'$.

%===========================
% 		HERE
%===========================

\begin{corollary}\label{coro emu dist}
	For any parameters $\epsilon <1$, $\kappa\geq 2$ and $1/\kappa< \rho< 0.5$, and any $n$-vertex graph $G=(V,E)$, our algorithm constructs a $\left(1+\epsilon,\beta \right)$-emulator with at most 
	$ \nfrac $
	edges in $O(\beta n^\rho)$ deterministic \congest\ time, where 
	$$\beta = \betadist.$$
\end{corollary}

Note that by setting $\kappa= f(n)\cdot ({\log n})$, for a function $f(n)= \omega(1)$, we obtain an emulator of size at most 
$n^{1+\frac{1}{f(n){\log n}}} = n+o(n). $
By Corollary \ref{coro emu}, we derive: 

\begin{corollary}\label{coro emu us dist}
	For any parameters $\epsilon <1$ and $ \rho< 0.5$, and any $n$-vertex graph $G=(V,E)$, our algorithm constructs a $\left(1+\epsilon,\beta \right)$-emulator with
	$ n+o(n) $
	edges in $O(\beta n^\rho)$ deterministic \congest\ time, where 
	$$\beta = 
	\left(\frac{ {\log (\rho{\log n})}+ \rho^{-1} }{\epsilon'\rho }\right)^{ {\log (\rho{\log n})}+ \rho^{-1}+O(1) } 
	.$$
\end{corollary}

\newcommand{\betadistspan}{\left(\frac{ {\log \kappa\rho}+ \rho^{-1} }{\epsilon\rho }\right)^{ {\log \kappa\rho}+ \rho^{-1}+O(1+{\log^{(3)}\kappa}) }}

\subsection{Fast Centralized Construction}
\label{sec fast cent}
To devise an efficient construction of ultra-sparse near-additive emulators in the centralized model of computation, one can simulate the construction provided in Section \ref{sec emu dist construction} in the centralized model. Given an unweighted, undirected graph $G=(V,E)$ on $n$ vertices, and parameters $\epsilon <1$, $\kappa = 1,2,\dots$ and $\rho \in [1/\kappa, 1/2]$, our distributed algorithm runs in $O(\beta\cdot n^\rho)$ time. Note that in every communication round, at most one message of $O({\log n})$ bits is sent along each edge of the graph $G$. Thus, simulating this algorithm in the centralized model can be done in $O(|E|\cdot \beta\cdot n^\rho)$ time. In fact, such a centralized implementation is simpler than the distributed construction. This is because, in the centralized model, there is no need to inform both endpoints of every emulator edge $(u,v)$ of the existence of the edge. Thus, constructing superclusters becomes much easier. Specifically, the execution of Task 3 is simpler, since there is no need to split trees of the forest $F_i$. 
The properties of the centralized construction are summarized in the following theorems.

\begin{theorem}\label{theo emu cent fast}
	For any parameters $\epsilon <1$, $\kappa\geq 2$ and $1/\kappa< \rho< 0.5$, and any $n$-vertex graph $G=(V,E)$, our algorithm deterministically constructs a $\left(1+\epsilon,\beta \right)$-emulator with at most 
	$ \nfrac $
	edges in $O(|E|\cdot \beta n^\rho)$ time in the centralized model of computation, where 
	$$\beta = \betadist.$$
\end{theorem}

Note that by setting $\kappa= f(n)\cdot ({\log n})$, for a function $f(n)= \omega(1)$, we obtain an emulator of size at most 
$n^{1+\frac{1}{f(n){\log n}}} = n+o(n). $
By Corollary \ref{coro emu}, we derive: 

\begin{theorem}\label{theo emu us cent fast}
	For any parameters $\epsilon <1$ and $ \rho< 0.5$, and any $n$-vertex graph $G=(V,E)$, our algorithm deterministically constructs a $\left(1+\epsilon,\beta \right)$-emulator with
	$ n+o(n) $
	edges in $O(|E|\cdot \beta n^\rho)$ time in the centralized model of computation, where 
	$$\beta = 
	\left(\frac{ {\log (\rho{\log n})}+ \rho^{-1} }{\epsilon'\rho }\right)^{ {\log (\rho{\log n})}+ \rho^{-1}+O(1) } 
	.$$
\end{theorem}

In fact, it is easy to see that the factor $\beta$ can be shaved off from this running time. We omit the details in this version of the paper. (The same is true concerning both Theorems \ref{theo emu cent fast} and \ref{theo emu us cent fast}.)

\section{Near-Additive Spanners}\label{sec span}
In this section, we show how one can modify the construction given in Section \ref{sec congest} to obtain sparse near-additive spanners. Specifically, given an unweighted, undirected graph $G=(V,E)$, and parameters $\epsilon>0$, $\kappa=  2,3,\dots$ and $\rho \in [1/k, 1/2]$, our current algorithm constructs a $(1+\epsilon,\beta)$-spanner with $O(\nfrac)$ edges, in $O(\beta n^\rho)$ deterministic \congest\ time, where 
\begin{equation}
\beta = \betadistspan. 
\end{equation}

We follow the construction described in Section \ref{sec emu dist construction}. We define $H= \emptyset$ and proceed in phases. Throughout the algorithm, instead of adding to $H$ an emulator edges $(u,v)$ with weight $d$, we add to the spanner $H$ a $u-v$ path from $G$ of length at most $d$. 
Recall that in the emulators construction, whenever a vertex $u\in V$ adds an edge $(u,v)$ with weight $d$ to the emulator, it sends a message to $v$ along a path from $G$, of weight at most $d$. In the current version of the algorithm, we will add to the spanner the entire path from $u$ to $v$, along which $u$ informs $v$ of the new edge. As a result, the construction of superclusters becomes simpler, because the message sent from $u$ to $v$ contains only the details of $v$, and does not need to contain any information regarding $u$ or the edge $(u,v)$. Therefore, there is no need to define hub-vertices as in Task 3 of the superclustering step (see Section \ref{sec super}). Thus, a single supercluster $\widehat{C}$ is formed from every tree $T$ in the forest $F_i$. Observe that Lemma \ref{lemma bouns pi1} (and as a result, \cref{eq explicit pi bound}) and Lemma \ref{lemma radpi leq ri dist} hold under this modification.

The distance threshold sequence remains as in Section \ref{sec emu dist construction}. 
To obtain sparse emulators, we adopt  the degree sequence used in \cite{ElkinN17}, and as a result, the number of phases of the algorithm slightly increases. The analysis of the number of edges in the current construction is closely related to the respective analysis in \cite{ElkinN17}.

Let $\gamma = {\max \{ 2, {\log {\log \kappa}} \}}$. Define $i_0 = {\min \{ \lfloor {\log {\gamma\kappa\rho}}\rfloor , \lfloor \kappa\rho\rfloor \}}$. For the exponential growth stage, which consists of phases $i\in [0,i_0]$, we set $deg_i = n^{\frac{2^i-1}{\gamma\kappa}+\frac{1}{\kappa}}$. Define $i_0+1$ as a \emph{transition phase}, and set $deg_{i_0+1} = n^{\rho/2}$. For the fixed growth stage, which consists of phases $i\in [i_0+1,\ell' = i_0+ \lceil 1/\rho-1/2\rceil]$, set $deg_i = n^\rho$. We will show that $|P_{\ell'}|\leq n^\rho$. Therefore, there are no popular clusters in the last phase, and the superclustering step can be safely skipped. 

By argument similar to those used in Section \ref{sec dist analysis of stretch} and \ref{sec rescale dist}, one can show that the additive term $\beta$ of such a construction is 
\begin{equation}
\beta = \betadistspan. 
\end{equation} 

In addition, by arguments similar to those used in Section \ref{sec dist analysis of rt} and \ref{sec rescale dist}, one can show that the running time of the algorithm can be upper-bounded by 
\begin{equation}
	O\left(\beta{n^\rho}\right).
\end{equation}

\subsection{Analysis of the Number of Edges}
In this section, we analyze the number of edges added to the spanner $H$ by every phase of the algorithm. We begin by analyzing the number of superclustering edges added to the spanner $H$ by every phase $i\in [0,\ell'-1]$ of the algorithm.

Observe that during each phase $i$, the superclustering edges added to the spanner $H$ all belong to a forest $F_i$. Hence, each phase contributes at most $n$ superclustering edges. Since there are $\ell'+1 = O({\log \gamma\kappa\rho} +1/\rho)$ phases, this implies that the the number of superclustering edges in the spanner $H$ is at most 
\begin{equation}\label{eq bound super edges}
O(n({\log \gamma\kappa\rho} +1/\rho)) = O(n({\log \kappa\rho} +1/\rho)) .
\end{equation}

\newcommand{\repsi}{\left(\frac{90\ell'}{\rho\epsilon}  \right)^i}
Next, we analyze the number of interconnection edges added to the spanner $H$ by every phase $i\in [0,\ell']$. Observe that interconnection edges are added only by centers $r_C$ of clusters $C\in U_i$. The center $r_C$ is charged with paths to all its neighboring clusters. Since $C\in U_i$, we know that $C$ is not popular. Therefore, it is charged with at most $deg_i$ paths.
In addition, by definition and by \cref{eq bound deltai}, the length of each such path is at most $\delta_i = O\left(\repsi\right)$. 
Hence, the number of edges charged to each center of a cluster in $U_i$ can be upper-bounded by $O\left(deg_i\cdot \repsi\right)$.

We restrict ourselves to the case where $\frac{90\ell'}{\rho\epsilon} \leq \frac{n^{\frac{1}{2\kappa}}}{2}$, which holds whenever $\kappa \leq \frac{c'{\log n}}{{\log (\ell'/(\rho\epsilon))}} $, for a sufficiently small constant $c'$. 
Observe that $U_i\subseteq P_i$. The number of interconnection edges added to the spanner $H$ by each phase $i$ can now be upper-bounded by 
 \begin{equation}\label{eq bound intercon edges}
 O\left(|P_i|\cdot deg_i\cdot \repsi\right)=O\left(|P_i|\cdot deg_i\cdot \left({n^{\frac{1}{2\kappa}}}/{2}\right)^i\right).
 \end{equation} 
 
In the next three lemmas, we bound the size of $P_i$ for the exponential growth stage, the transition phase, and the fixed growth stage, respectively. 

\begin{lemma}
	\label{lemma bound pi exp}
	For $i\in [0, i_0+1]$, we have 
	$| {P}_i| \leq n^{1 - \frac{2^i-1-i}{\gamma\kappa}-\frac{i}{\kappa}}.$
\end{lemma}

\begin{proof}
We will prove the lemma by induction on the index of the phase $i$. 

For $i=0$, the right-hand side is $n^{1 - \frac{2^0-1-0}{\gamma\kappa}-\frac{0}{\kappa}} = n$. Thus the claim is trivial. 

Assume that the claim holds for some $i\in [0,{\ell'}-1]$ and prove it also holds for $i+1$.
By \cref{eq explicit pi bound} we have that $|P_{i+1}| \leq |P_i|\cdot deg_i^{-1}$.

Together with the induction hypothesis, and since for i, $i \in[0,i_0]$, we have $deg_i = n^{\frac{2^i-1}{\gamma\kappa}+\frac{1}{\kappa}}$, we have that 
\begin{equation*}
\begin{array}{lclclclclclc}
|P_{i+1} | &\leq& |P_i|\cdot deg_i^{-1} 
&\leq & n^{1 - \frac{2^i-1-i}{\gamma\kappa}-\frac{i}{\kappa}} \cdot n^{-\frac{2^i-1}{\gamma\kappa}-\frac{1}{\kappa}}

&= & n^{1 - \frac{2^{i+1}-1-(i+1)}{\gamma\kappa}-\frac{i+1}{\kappa}} .
\end{array}
\end{equation*}
\end{proof}

Recall that $\gamma \geq 2$. Observe that by \cref{eq bound intercon edges} and Lemma \ref{lemma bound pi exp}, we have that the number of edges added to the spanner $H$ by every phase $i\in [0,i_0]$ is at most 

 \begin{equation}\label{eq bound intercon edges exp}
 \begin{array}{lclclclc}
O\left(|P_i|\cdot deg_i\cdot \left({n^{\frac{1}{2\kappa}}}/{2}\right)^i\right)
&=& 
O\left( n^{1 - \frac{2^i-1-i}{\gamma\kappa}-\frac{i}{\kappa}}\cdot 
n^{\frac{2^i-1}{\gamma\kappa}+\frac{1}{\kappa}}
\cdot {n^{\frac{i}{2\kappa}}}\cdot 2^{-i}\right)
%&=& 
%O\left(n^{1 - \frac{2^i-1}{\gamma\kappa} +\frac{i}{\gamma\kappa} -\frac{i}{\kappa} +\frac{2^i-1}{\gamma\kappa}+\frac{1}{\kappa}+ \frac{i}{2\kappa}} \cdot {2}^{-i}\right)\\
&=& 
O\left( {2}^{-i} n^{1 +\frac{1}{\kappa}} \right).
 \end{array}
\end{equation}

\begin{lemma}
	\label{lemma bound pi tran}
	The size of the input collection $P_{i_0+1}$ for the transition phase satisfies
	$| {P}_{i_0+1}| \leq n^{1 - \rho}.$
\end{lemma}

\begin{proof}
	By Lemma \ref{lemma bound pi exp}, we have 
	$$| {P}_{i_0+1}| \leq n^{1 - \frac{2^{i_0+1}-1-({i_0+1})}{\gamma\kappa}-\frac{{i_0+1}}{\kappa}}= n^{1 -
	\frac{2^{i_0+1}-1}{\gamma\kappa} +\frac{i_0+1-\gamma({i_0+1})}{\gamma \kappa}}= 
	n^{1 -\frac{2^{i_0+1}-1}{\gamma\kappa} -\frac{(i_0+1)(\gamma-1)}{\gamma \kappa}}.$$
	Observe that since $i_0>0$ and $\gamma\geq2$, we have 
	$\frac{(i_0+1)(\gamma-1)}{\gamma \kappa} \geq \frac{1}{\gamma \kappa}$. 
	
	If $i_0 = \lfloor {\log {\gamma\kappa\rho}}\rfloor$, then 
	\begin{equation}
	| {P}_{i_0+1}| \leq 
	n^{1 -\frac{2^{\lfloor {\log {\gamma\kappa\rho}}\rfloor+1}-1}{\gamma\kappa} 
		-\frac{(i_0+1)(\gamma-1)}{\gamma \kappa}}
	\leq 
	n^{1 -\frac{ \gamma\kappa\rho-1}{\gamma\kappa} -\frac{1}{\gamma \kappa}}
	=
	n^{1 -\rho}.
	\end{equation}
	
	Otherwise, if $i_0 = \lfloor \kappa\rho\rfloor$, then 	
	\begin{equation}
	| {P}_{i_0+1}| \leq 
	n^{1 -\frac{2^{i_0+1}-1}{\gamma\kappa} -\frac{\kappa\rho(\gamma-1)}{\gamma \kappa}}
	=
	n^{1 -\frac{2^{i_0+1}-2}{\gamma\kappa} -\frac{\gamma\kappa\rho}{\gamma \kappa}}
	\leq 
	n^{1 -\rho}.	
	\end{equation}	
\end{proof}

Recall that $i_0 \leq \lfloor \kappa\rho\rfloor $. Observe that by \cref{eq bound intercon edges} and Lemma \ref{lemma bound pi tran}, we have that the number of edges added to the spanner $H$ by phase ${i_0+1}$ is at most 

\begin{equation}\label{eq bound intercon edges tran}
\begin{array}{lclclclc}
O\left(|P_{i_0+1}|\cdot deg_{i_0+1}\cdot \left({n^{\frac{1}{2\kappa}}}/{2}\right)^{i_0+1}\right)
&= & O\left({2}^{-(i_0+1)}n^{1 - \frac{\rho}{2} +\frac{\kappa\rho+1}{2\kappa}}\right)
&= & O\left({2}^{-(i_0+1)}n^{\frac{1}{2\kappa}}\right).
\end{array}
\end{equation}

\begin{lemma}
	\label{lemma bound pi fixed}
	For every $j\in [2,{\ell'}-i_0]$ we have $|P_{i_0+j}| \leq n^{1-\rho/2 - (j-1) \rho}$.
\end{lemma}
\begin{proof}
	The proof is by induction on the index $j$. For $j=2$, by Lemma \ref{lemma bound pi tran} we have
	$|{P}_{i_0+1}| \leq n^{1 - \rho}$. In addition, by \cref{eq explicit pi bound} we have that $|P_{i_0+2}| \leq |P_{i_0+1}|\cdot deg_{i_0+1}^{-1}$. Recall that $deg_{i_0+1} = n^{\rho/2}$. Hence we have $ |{P}_{i_0+1}| \leq n^{1 - \rho -\rho/2}$ and so the claim holds.

	Assume that the claim holds for some $j\in [2,{\ell'}-i_0-1]$ and prove it holds for $j+1$. Recall that $deg_{i_0+j} = n^\rho$. By \cref{eq explicit pi bound} we have that $|P_{i_0+j+1}| \leq |P_{i_0+j}|\cdot deg_{i_0+j}^{-1}$. Together with the induction hypothesis, we have 
	
	\begin{equation*}
	\begin{array}{lclclclclclc}
		|P_{i_0+j} | &\leq&
		 n^{1-\rho/2 - (j-1)\cdot \rho}\cdot n^{-\rho} &=& n^{1-\rho/2 - (j+1-1)\rho} . 
	\end{array}
	\end{equation*}
\end{proof}

Recall that $\ell' = i_0+ \lceil 1/\rho-1/2\rceil$. Hence, by Lemma \ref{lemma bound pi fixed} we have 
\begin{equation}
\begin{array}{lclclcl}	
|P_{\ell'}| &\leq& n^{1-\rho/2 - (\lceil 1/\rho-1/2\rceil-1)\cdot \rho} &\leq&
n^{1-\rho/2 - 1+ (3/2)\cdot \rho}& = & n^{\rho}.
\end{array}
\end{equation}
Thus, there are no popular clusters during phase $\ell'$, and \cref{eq bound intercon edges} holds also for the last phase.

Recall that $i_0\leq \lfloor \kappa\rho \rfloor$, and therefore $\frac{i_0}{2\kappa} \leq \frac{\rho}{2}$.
It follows that 
	\begin{equation*}
\begin{array}{lclclclclclc}
\left(\frac{n^{\frac{1}{2\kappa}}}{2}\right)^{i_0} 
& = & 2^{-i_0}n^{\frac{i_0}{2\kappa}}
& \leq & 2^{-i_0}n^{\frac{\rho}{2}}.
\end{array}
\end{equation*}

 Also, recall that $1/\kappa <\rho$. 
Observe that by \cref{eq bound intercon edges} and Lemma \ref{lemma bound pi fixed}, we have that the number of edges added to the spanner $H$ by every phase $i_0+j$, for $j\in [2,{\ell'}-i_0]$ is at most

\begin{equation}\label{eq bound intercon edges fixed}
\begin{array}{lclclclc}
O\left(|P_{i_0+j}|\cdot deg_{i_0+j}\cdot \left({n^{\frac{1}{2\kappa}}}/{2}\right)^{i_0+j}\right)
&=&

O\left( n^{1-\rho/2 - (j-1)\cdot \rho} \cdot n^{\rho} \cdot n^{\rho/2}\cdot 2^{-i_0}\left({n^{\frac{1}{2\kappa}}}/{2}\right)^{j}\right)

\\
&=& 

O\left(2^{-i_0 -j} n^{1 - (j-2)\cdot \frac{1}{\kappa}+ \frac{j}{2\kappa}}\right)

\\

&=& 

O\left(2^{-i_0 -j} n^{1 + \frac{1}{\kappa}}\right).

\end{array}
\end{equation} 

By \cref{eq bound super edges,eq bound intercon edges exp,eq bound intercon edges tran,eq bound intercon edges fixed} we have that the overall number of edges added to the spanner $H$ by all phases of the algorithm is

\begin{equation}
O(n({\log\kappa\rho} +\frac{1}{\rho})) + \sum_{i = 0}^{i_0} O\left( {2}^{-i} n^{1 +\frac{1}{\kappa}} \right) + O\left({2}^{-(i_0+1)}n^{\frac{1}{2\kappa}}\right) + \sum_{j=2}^{{\ell'}-i_0} O\left(2^{-i_0-j} n^{1 + \frac{1}{\kappa}}\right) = O\left(\nfrac\right).
\end{equation}

Recall that we restrict ourselves to the case where $\kappa \leq \frac{c'{\log n}}{{\log (\ell'/(\rho\epsilon))}} $, for a sufficiently small constant $c'$. Also recall that $\ell' \leq \lfloor {\log \gamma\kappa\rho}\rfloor +  \lceil1/\rho-1/2\rceil$. 
Note that $\frac{c'{\log n}}{{\log (\ell'/(\rho\epsilon))}} \geq \frac{\Omega({\log n})}{{\log(1/\epsilon)} + {\log (1/\rho)}+{\log ^{(3)}n}}$, where ${\log ^{(3)}n}$ is the three-times iterated logarithm. 
The following corollary summarizes the properties of current construction.

\begin{corollary}\label{coro span}
	For any unweighted, undirected $n$-vertex graph $G=(V,E)$, and any parameters $\epsilon<1$, $\kappa \in [2,\frac{c{\log n}}{{\log(1/\epsilon)} + {\log (1/\rho)}+{\log ^{(3)}n}}] $, for a constant $c$ and $\rho \in [1/\kappa,1/2]$, our algorithm computes a $(1+\epsilon,\beta)$-spanner with $O(\nfrac)$ edges in $O(\beta n^\rho)$ deterministic \congest\ time, where 
	$$\beta = \betadistspan.$$
\end{corollary}

To obtain the sparsest spanners that one can get with this construction, we  set $\epsilon>0$ to be an arbitrarily small constant, and $\kappa =   \frac{c'{\log n}}{{\log ^{(3)}n}} $. Under this assignment of parameters, the size of the spanner is just $O(n{\log {\log n}})$, and the additive error is  $\beta = O(\frac{{\log {\log n}}+1/\rho}{\rho})^{{\log {\log n}}+1/\rho}$. 
	
	\bibliographystyle{alpha}

\bibliography{emu_cite}

\end{document}